%% file: journal_v6.tex
\documentclass[journal,onecolumn,11pt]{IEEEtran}
\usepackage{subfigure,cite,graphicx,amsmath,amssymb,eufrak,mathrsfs,epsf,epsfig}
\usepackage[usenames]{color}
\usepackage{psfrag}

\usepackage[active]{srcltx}
\ifCLASSINFOpdf

\else
\fi

\input{macros}

\begin{document}
\newtheorem{ach}{Achievability}
\newtheorem{con}{Converse}
\newtheorem{definition}{Definition}
\newtheorem{theorem}{Theorem}
\newtheorem{lemma}{Lemma}
\newtheorem{example}{Example}
\newtheorem{cor}{Corollary}
\newtheorem{prop}{Proposition}
\newtheorem{conjecture}{Conjecture}
\newtheorem{remark}{Remark}
\title{Wireless Multihop Device-to-Device Caching Networks}
\author{\IEEEauthorblockN{Sang-Woon Jeon,~\IEEEmembership{Member,~IEEE}, Song-Nam Hong,~\IEEEmembership{Member,~IEEE}, Mingyue Ji,~\IEEEmembership{Member,~IEEE},\\
Giuseppe Caire,~\IEEEmembership{Fellow,~IEEE}, and Andreas F. Molisch,~\IEEEmembership{Fellow,~IEEE}}
\thanks{This work was supported in part by the Basic Science Research Program through the National Research Foundation of Korea (NRF) funded by the Ministry of Education, Science and Technology (MEST) [NRF-2013R1A1A1064955] and by the NSF Grants CCF 1161801 and EARS  1444060.}
\thanks{The material in this paper was presented in part at the IEEE Information Theory Workshop (ITW), Jerusalem, Israel, April 2015 and at the IEEE International Conference on Communications (ICC), London, United Kingdom, June 2015.}
\thanks{S.-W. Jeon is with the Department of Information and Communication Engineering, Andong National University, Andong, South Korea (e-mail: swjeon@anu.ac.kr).}%
\thanks{S.-N. Hong is with the Ericsson Research Lab., San Jose, CA 95118, USA (e-mail: songnam.hong@ericsson.com).}%
\thanks{M. Ji and A. F. Molisch are with the Ming Hsieh Department of Electrical Engineering,
University of Southern California, Los Angeles, CA 90089, USA (e-mail: mingyuej@usc.edu; molisch@usc.edu).}%
\thanks{G. Caire is with the Department of Telecommunication Systems, Technical
University of Berlin, Berlin 10623, Germany, and also with the Ming Hsieh
Department of Electrical Engineering, University of Southern California,
Los Angeles, CA 90089, USA (e-mail: caire@tu-berlin.de).}
}
\maketitle

\begin{abstract}
We consider a wireless device-to-device (D2D) network where $n$ nodes 
are uniformly distributed at random over the network area. 
We let each node with storage capacity $M$ cache files from a library of size $m \geq M$.
Each node in the network requests a file from the library independently at random, according to a popularity distribution, and
is served by other nodes having the requested file in their local cache via (possibly) multihop transmissions.
Under the classical ``protocol model'' of wireless networks,
we characterize the optimal per-node capacity scaling law for a broad class of heavy-tailed popularity distributions
including Zipf distributions with exponent less than one.
In the parameter regimes of interest, we show that a decentralized random caching strategy
with uniform probability over the library yields the optimal per-node capacity scaling of $\Theta(\sqrt{M/m})$, which is
constant with $n$, thus yielding throughput scalability with the network size.
Furthermore, the multihop capacity scaling can be significantly better than 
for the case of single-hop caching networks, for which the per-node capacity is $\Theta(M/m)$.
The multihop capacity scaling law can be further improved for a Zipf distribution with exponent 
larger than some threshold $> 1$, by using a decentralized random caching uniformly across a subset of most popular 
files in the library. Namely, ignoring a subset of less popular files (i.e., effectively reducing the size of the library)
can significantly improve the throughput scaling while guaranteeing that all nodes will be served with high probability as $n$ increases.
\end{abstract}

 \IEEEpeerreviewmaketitle

\section{Introduction}  \label{intro}

Internet traffic has grown dramatically in recent years, mainly due to on-demand video streaming  \cite{CISCO}.
While wireless is by far the preferred way through which users connect to the Internet, today's cellular technology and
service providers do not support seamless cost-effective on-demand video streaming. For example, most monthly cellular data plans
would be completely consumed by a {\em single} streaming session of a standard definition movie from a typical services such as
Netflix, iTune, or Amazon Prime (duration ~1h:30, size ~2GB). It is evident that in order to fill in the gap between the users' expectation
and the limitations of the provided services, a dramatic technology paradigm shift is required.
In this perspective, it has been recently recognized that {\em caching at the wireless edge}, i.e., caching the content library directly in the wireless
nodes (femtocell base stations or user devices), has the potential of solving the problem of network scalability by providing per-node
throughput that scales much better than conventional unicast transmission, in a variety of scenarios.

One important feature of on-demand video streaming is that user demands are highly redundant over time and space.
As an example, consider a university campus where $n \approx 10000$ users (distributed over a surface of $\approx 1 km^2$) stream movies from a library of $\approx 100$ files, such as the weekly top-of-the chart titles of Netflix, iTune, or Amazon Prime. For such scenario, each user demand can be satisfied by local communication from a cache, without cluttering a cellular base station with thousands of unicast sessions, or without requiring to deploy
a large number of small cell access points, each requiring  costly high-throughput backhaul.
Intuitively, caching can effectively take advantage of the inherent redundancy of the user demands, although, differently from live streaming,
in on-demand streaming users do not request the same content at the same time (this type of redundancy is referred to in
\cite{shanmugam2013,ji2013optimal,ji2013throughput} as {\em asynchronous content reuse}).

\subsection{Related Work}

\subsubsection{Conventional ad-hoc networks}

Since the seminal work of Gupta and Kumar \cite{GuptaKumar:00},
the capacity scaling laws of wireless ad-hoc networks has been extensively studied
(e.g., \cite{kulkarni2004deterministic,li2009multicast,franceschetti2007closing}).
The model introduced in \cite{GuptaKumar:00} consists of $n$ nodes placed uniformly at random on a planar region
and grouped into source--destination (SD) pairs at random.
Assuming an interference avoidance constraint referred to as the \emph{protocol model} (see Section \ref{sec:PF}),
it was shown in \cite{GuptaKumar:00} that the per-node capacity must scale as $O(\frac{1}{\sqrt{n}})$ (i.e., upper bound).
Furthermore, a simple ``straight-line" multihop relaying scheme achieves the per-node throughput scaling of $\Omega(\frac{1}{\sqrt{n \log n}})$.
The same results were confirmed in \cite{kulkarni2004deterministic} by using a simpler and more general analysis technique.
Later, the $1/\sqrt{\log n}$ gap factor between converse and achievability was closed in \cite{franceschetti2007closing},
by showing that the per-node throughput scaling of $\Theta(\frac{1}{\sqrt{n}})$ is indeed achievable by using a more refined multihop strategy
based on percolation theory.

Beyond the protocol model, the capacity scaling law of wireless ad-hoc networks has been also studied in an information theoretic sense, considering
a {\em physical model} that includes distance-dependent propagation path-loss, fading, Gaussian noise, and signal interference
(e.g., \cite{xie2004network,xue2005transport,shakkottai2010multicast, niesen2010balanced,ozgur2007hierarchical}).
While the protocol model is scale-free, the physical model behaves differently depending on whether the network is
``extended'' (constant node density, with the network area growing as $\Theta(n)$), or ``dense'' (constant network area, with the node density growing as
$\Theta(n)$).  In \cite{xie2004network, xue2005transport}, the achievability schemes are based on multihop strategy, point-to-point coding, and treating interference as (Gaussian) noise.
For the extended network model, it was shown that if the path-loss exponent is greater than or equal to three, then the scaling law is the same as for the protocol model and the multihop strategy is sufficient to achieve
the optimal scaling. In contrast, for the extended network model with the path-loss exponent less than three and for the dense network model (in this case the path-loss exponent is irrelevant) the multihop strategy is suboptimal. In these cases, the {\em hierarchical cooperation} scheme
proposed in \cite{ozgur2007hierarchical} (see also improved and optimized hierarchical cooperation scheme in \cite{ozgur2010,hong2015beyond})
achieves an almost optimal throughput scaling within a factor of $n^{\epsilon}$, where $\epsilon$ can
be made arbitrarily small as the number of hierarchical stages increases.

\subsubsection{Caching networks}

Motivated by the considerations made at the beginning of this section,
wireless caching networks have been the subject of recent intensive investigation
\cite{shanmugam2013,Gitzenisr:13,
ji2013optimal,ji2013throughput,JiCaireMolisch2013,ji2013fundamental,maddah2013fundamental,maddah2014decentralized,niesen2013coded,ji2015random,ji2014caching, karamchandani2014hierarchical,hachem2014multi}.
The single-hop device-to-device (D2D) case was considered in \cite{ji2013optimal,ji2013throughput,JiCaireMolisch2013},
where $n$ user node request files from a library of $m$ files
according to a common demand or popularity distribution and each node has cache capacity constraint equal to the size
of $M \leq m$ files. The delivery scheme (i.e., the coordination of transmissions in order to serve the users' requests)
is restricted to be one-hop, i.e., either the requested file is found in the cache, or it is directly downloaded from a neighbor node
through a D2D wireless link. Under a Zipf popularity distribution \cite{Breslau:99} with parameter less than one and the protocol model of \cite{GuptaKumar:00},
it was shown in \cite{ji2013optimal,ji2013throughput,JiCaireMolisch2013} that the per-node throughput scales as $\Theta{(M/m)}$. This can be
achieved by an independent and random caching placement and a TDMA-based link scheduling scheme, at the expense of a positive outage probability,
due to the random nature of the caching placement scheme. However, in the relevant regime where $nM \gg m$, this outage probability can be kept under control,
i.e., the system can be designed in order to achieve any target outage probability $\epsilon > 0$, for sufficiently large $n$.
It is remarkable to notice that the per-node throughput in this case scales much better than
in the case of general ad-hoc networks under the protocol model. In fact, while in the general case the per-node throughput converges to zero
with the size of the network as $1/\sqrt{n}$, here it is constant with $n$ and directly proportional to the fraction of cached files $M/m$.
This much better scaling can be explained as an effect of the dense spatial spectrum reuse allowed by caching,
for which the requested content is found within a short communication radius, and therefore
a large number of simultaneous D2D links can be active on the same time slot.
Furthermore, an information theoretic study of the one-hop D2D caching network in the case of
worst-case arbitrary demands is provided in \cite{ji2014fundamental}, where the same throughput scaling
of $\Theta{(M/m)}$ is achieved through {\em inter-session network coded multicasting only} scheme, {\em spatial reuse only} scheme without inter session coding
as in \cite{ji2013optimal}, or a combination of both schemes.

A different  one-hop caching network topology has been studied  in~\cite{maddah2013fundamental,maddah2014decentralized,niesen2013coded,ji2015random}, where a single transmitter (i.e., a base station with all files in the library) serves $n$ user nodes through a common noiseless
link of fixed capacity (bottleneck link). The scheme proposed in \cite{maddah2013fundamental, maddah2014decentralized} partitions each file into packets  and each node stores subsets of packets from each file. This provides ``side information'' at each node such that, for the worst-case demands setting,
the base station can compute a multicast {\em network-coded} messages (transmitted via the common link) such that each node can decode
its own requested file from the multicast message and its cached side information.
Also in this case, the per-node throughput scaling under the worst-case arbitrary demands model
is again given by $\Theta{(M/m)}$, which is remarkably identical with the throughput scaling
achieved by single-hop D2D caching networks. In this case, the caching gain is explained in terms of  ``coded multicasting gain'', i.e., in the ability of turning unicast traffic into coded multicast traffic, such that one transmission satisfies multiple nodes. Further, when the user demands are random and
follow a Zipf distribution,  the order optimal average rate was characterized in \cite{ji2015random}. This
behaves as a function of all the system parameters including the number of users, the library size, the memory size and the popularity distributions.
Remarkably, in all the regimes of system parameters, the cache memory size $M$ can provide a {\em multiplicative gain}, which can be linear, sub-linear, or super-linear, depending on the cases.  A number of extensions, such as multiple number of requests, hierarchical network structures, and extension to multiple servers under various topology assumptions, can be found in \cite{ji2014caching,ji2015caching,karamchandani2014hierarchical,hachem2014multi,ji2015combination,shariatpanahi2015}.

\subsection{Contributions}
\label{sec: Contributions}
In this paper, we study a natural extension of the single-hop D2D network by allowing multihop transmission.
As a related work, a multihop transmission scheme for wireless caching networks has been studied in \cite{Gitzenisr:13} under the protocol model.
The key differences between the present paper and \cite{Gitzenisr:13} are as follows.
First, the main objective of \cite{Gitzenisr:13} is to minimize the average number of flows passing through each node. Such average number of
flows is proportional to the reciprocal of the average per-node throughput only for certain network model;
on the other hand, we directly derive the optimal scaling law of the per-node throughput.
Second, a centralized and deterministic caching placement was proposed in \cite{Gitzenisr:13} according to the popularity distribution;
in contrast, we present a completely decentralized random caching placement according to a uniform distribution over the whole
 file library, which is ``universal'' since it is independent of the specific popularity distribution.
Remarkably, while the placement and the achievability scheme of  \cite{Gitzenisr:13} would break under a node layout permutation,
such that one should re-allocate the cache content when the nodes are in the presence of node mobility, our scheme is robust since
any random permutation of the nodes would generate the same caching distribution, and therefore yields the same throughput scaling
with high probability.
Third, the file delivery scheme in  \cite{Gitzenisr:13} allows for multihop SD paths 
(i.e., between nodes caching a given file and nodes demanding such file) 
of the order of $\sqrt{n}$, i.e., the delivery paths are allowed to traverse the whole network.
In contrast,  in this paper we consider a more practical achievability scheme called {\em local multihop protocol}, where the number of
hops between any SD pairs are independent of the number of nodes and decreases when the storage capacity per node increases. 

The proposed caching placement and delivery scheme yield a per-node capacity scaling of $\Theta{(\sqrt{M/m})}$,
which is order-optimal when the popularity distribution has the ``heavy tail" property (see Definition \ref{def:distribution_class} in Section \ref{subsec:throughput}). For example, this is the case of a Zipf distribution with exponent less than one \cite{Breslau:99}.\footnote{Throughout the paper, an ``order-optimal''
scheme  means that it achieves the optimal throughput scaling law within a multiplicative gap of $n^{\epsilon}$ for any $\epsilon>0$.}
This result shows that multihop yields a much better per-node capacity scaling than single-hop D2D networks, which is given by $\Theta(M/m)$.
Furthermore, we show that for other popularity distributions, where the ``heavy tail" property is not satisfied or the user demands strongly concentrate,
a further improvement of the per-node throughput scaling beyond $\Theta{(\sqrt{M/m})}$ is achievable, similar to the case of single-hop D2D networks in \cite{Gitzenisr:13, ji2015random}.

\subsection{Paper Organization}

In Section~\ref{sec:PF}, we provide our network model and some definitions to be used throughout the paper. Section~\ref{sec:main} states the main results of this paper on the per-node capacity scaling laws for caching wireless D2D networks. In Section~\ref{sec:scheme}, we present an achievable scheme which is universal independently from a popularity distribution. An upper bound is provided in Section~\ref{sec:upper_bound}. In Section~\ref{sec:improved_throughput}, we further improve the throughput scaling laws for a Zipf distribution with exponent larger than a certain threshold. Some concluding remarks are provided in Section~\ref{sec:conc}.

\section{Problem Formulation}  \label{sec:PF}

In this section, we provide the model of the network under investigation and
define achievable throughput and system scaling regimes.
Generally speaking, for a sequence of events $\{E_n : n = 1,2, 3,\ldots\}$ we say that $E_n$ ``occurs with high probability'' (whp) 
if $\lim_{n \rightarrow \infty} \PP(E_n) = 1$, where it is understood that these events are defined in an appropriate probability space, with probability measure 
generally indicated by $\PP(\cdot)$.  
For notational convenience, let $\overset{\text{whp}}{\geq}$ and $\overset{\text{whp}}{\leq}$ denote that the corresponding
inequalities hold whp. We will also use the following order notations \cite{Knuth:76}.
\begin{itemize}
\item $f(n)=\operatorname{O}(g(n))$ if there exist $c>0$ and $n_0>0$ such that $f(n)\leq c g(n)$ for all $n\geq n_0$.
\item $f(n)=\Omega(g(n))$ if $g(n)=\operatorname{O}(f(n))$.
\item $f(n)=\Theta(g(n))$ if $f(n)=\operatorname{O}(g(n))$ and $g(n)=\operatorname{O}(f(n))$.
\end{itemize}

\subsection{Caching in Wireless Multihop D2D Networks}

We consider a wireless multihop D2D network consisting of a population $\Uc$ of $n = |\Uc|$ nodes, distributed uniformly 
and independently over a unit square area $[0,1] \times [0,1]$. 
Let $d(u,v)$ denote the distance between nodes $u,v \in \Uc$.
It is assumed that communication between nodes follows the {\em protocol model} of \cite{GuptaKumar:00}: the transmission from
node $u$ to node $v$ is successful if and only if: i) $d(u,v) \leq r$, and ii) no other active transmitter must be in a circle of radius
$(1+\Delta)r$ from the receiver node $v$. Here, $r,\Delta>0$ are given protocol parameters.
Also, each node sends its packets at some constant rate $W$ bits/s/Hz.

We consider a library ${\cal F} = \{W_1, \cdots, W_m\}$ of $m = |\Fc|$ files (information messages), such that messages $W_f$ 
are drawn at random and independently with a uniform distribution over a message set $\FF_2^B$ (binary strings of length $B$), 
for some arbitrary integer $B$. It follows that each file in $\Fc$ has entropy $H(W_f) = B$ bits.
Consistently with the current information theoretic literature on caching networks (see Section \ref{intro}),
a {\em caching scheme} is formed by two phases: caching placement and delivery.  
The file library is generated, and then maintained fixed for a long time. Each network node (user) has an on-board cache memory 
of capacity $MB$ bits, i.e., expressed in ``equivalent file-size'' the cache capacity is equal to $M$ files. 
The problem consists of storing information in the caches such that the delivery is as efficient as possible.
It is important to note that the caching placement phase is performed beforehand, when the file library is generated. Then, 
each node $u \in \Uc$ demands a file with index $f_u \in \{1, \cdots, m\}$, 
and the network must coordinate transmissions
(in particular, in this paper we consider multihop D2D operations according to the above defined protocol model), such that each 
demand is satisfied, i.e., each user $u$ is able to decode its desired files $f_u$ from the content of its own cache and from what it receives 
from the other nodes.

In general, the caching phase is defined by a collection of $n$ maps $Z_u : \FF_2^{Bm} \rightarrow \FF_2^{BM}$, 
such that $Z_u(\Fc)$ is the content of the cache at node $u \in \Uc$. Notice that 
the  cache content  is independent of the 
demand vector $(f_1, \cdots, f_n)$, reflecting the fact that the caching phase is performed beforehand. 
In this sense, the caching placement can be regarded as part of the ``code set-up''.
In the achievability strategies considered in this paper we consider only 
caching of entire files ($M$ files per node).  As a result, as in \cite{shanmugam2013,ji2013throughput,Gitzenisr:13},
the parameter $B$ (file size) is irrelevant for our achievability results.\footnote{Notice that this is not the case for other schemes such as in \cite{ji2014fundamental,maddah2013fundamental,maddah2014decentralized}, 
where the file size plays an important role (see \cite{shanmugam2014}.} 

Restricting caching to entire files, a caching placement realization is uniquely defined by a bipartite graph
$\Gsf = (\Uc, \Fc, \Ec)$  with  ``left'' nodes $\Uc$, ``right'' nodes $\Fc$ and edges $\Ec$ 
such that $(u,f) \in \Ec$ indicates that file $W_f$ is assigned to the cache of node $u$.
A bipartite cache placement graph $\Gsf$ is feasible if the degree of each node $u \in \Uc$ is not larger than the cache constraint $M$.
Let $\Gc$ denote the set of all feasible bipartite graphs $\Gsf$. 
Then, we define a random cache placement as a probability mass function $\Pi_c$ over $\Gc$. In particular, if $\Pi_c$ is induced by 
randomly and independently assigning $M$ files to each user node $u \in \Uc$, we say that the cache placement is 
``decentralized''. For a decentralized caching placement, each user node chooses its own $M$ files independently of the other nodes. 

After the caching functions are computed and the result is stored in the user nodes' caches, 
the network is repeatedly used in {\em rounds}. At each round, each node requests a file in the library, and the network 
must satisfy such requests. Since the network resets itself at the end of each delivery cycle, by the renewal--reward theorem \cite{Tijms:03}
the per-node throughput is given by the reciprocal of the time needed to deliver the files
(up to a multiplicative constant that depends on $W$, on the system bandwidth, and  on the file size).
Two models for the user demands have been investigated in the literature:
{\em arbitrary} and {\em random}. In the first case, the users' demand vector $(f_1, \cdots, f_n)$ is arbitrary, and the delivery time is defined for the 
worst-case demand configuration \cite{maddah2013fundamental,maddah2014decentralized,ji2013fundamental}. 
In the second case, the demands are generated at random and the delivery time is averaged over the users' demand distribution
\cite{shanmugam2013,ji2013optimal,ji2013throughput,ji2015random}. In this paper, we consider the random demands setting. In particular, we assume that
the users' demands are independently and uniformly distributed according to a common probability mass function $\{p_{r}(f): f \in \{1, \cdots, m\}\}$. 
The probability mass function $p_{r}(\cdot)$ is referred to in the following as the {\em popularity distribution}.
Without loss of generality, we assume a descending order between request probabilities, i.e, $p_{r}(i)\geq p_{r}(j)$ if $i\leq j$ for $i,j\in \{1, \cdots, m\}$.
For instance, a Zipf popularity distribution with exponent $\gamma>0$ is defined by $p_{r}(i) = \frac{i^{-\gamma}}{\sum_{j=1}^m j^{-\gamma}}$ 
for $i \in \{1, \cdots, m\}$ \cite{Breslau:99}. 

In the following, all events regarding a network of size $n = 1, 2, 3, \ldots$ are defined on a common probability space generated by the
random placement of the nodes, indicated by $\Psf$, the random placement of the caches, indicated by $\Gsf$, and the random 
demand vector, indicated by $\fsf$.

\subsection{Achievable Throughput and System Scaling Regime} \label{subsec:throughput}

In order to study capacity scaling of the caching wireless multihop D2D network defined before, we consider $m$ and $M$ expressed as functions of $n$ as
\begin{equation} \label{eq:relation_parameters}
m = a_1 n^{\alpha} \mbox{ and } M  =  a_2 n^{\beta},
\end{equation}
where $\alpha,a_1,a_2> 0$ and $\beta\in[0,\alpha]$.
We assume that $a_1>a_2$ if $\alpha=\beta$ because the delivery phase becomes trivial if $\alpha=\beta$ and $a_1\leq a_2$ (each node is able to store the 
entire library $\mathcal{F}$ for this case).

Before entering the analysis, it is important to clearly define the concept of outage event and symmetric throughput. 
For a given node placement $\Psf$, cache placement $\Gsf$, and demand vector $\fsf$, a feasible delivery strategy consists of a sequence of 
activation sets, i.e., sets of active transmission links, $\{\Ac_t : t = 1, 2, 3, \ldots \}$, such that 
at each time $t$ the active links in $\Ac_t$ do not violate the protocol model. For a given feasible delivery strategy, we let 
$T_n$ denote the corresponding per-node symmetric throughput, i.e., the rate (in bit/s/Hz) at which the 
request of any node in the network can be served with vanishing probability of error, as $B \rightarrow \infty$.
If for some node $u \in \Uc$ the message probability of error is lower bounded 
by some positive constant for all $B$, we say that the network is in outage. In this case, conventionally, we let $T_n = 0$.  

A sufficient condition for outage is that there exists some $u \in \Uc$ for which $W_{f_u}$ cannot be 
reconstructed from the whole cache content $\{ Z_v : v \in \Uc\}$.  Within the assumptions of our model, 
it is easy to see that the above condition also necessary. In fact, by contradiction, notice that 
if for all $u \in \Uc$ the requested message $W_{f_u}$ can be reconstructed from $\{ Z_v : v \in \Uc\}$,  
then there exists some delivery strategy that conveys all the cache messages to all
the user nodes by an appropriate multihop schedule, such that all nodes can decode their own desired file. 
This is an immediate consequence of the fact that the transmission in any single active link of the network is error-free, 
and that any node can communicate with any other node, by letting the transmission radius $r$ sufficiently large. 
Of course, conveying the global cache content to all nodes may take a very long delivery time, yielding low throughput. As a matter of fact,  
studying the behavior of the optimal $T_n$ as $n \rightarrow \infty$ is precisely the goal pursued in the rest of this paper. 

From what said above, $T_n$ is a random variable, function of $\Psf$, $\fsf$, and $\Gsf$. In general, the cumulative distribution function of $T_n$ takes on the form:
\[ F_{T_n}(x) =  \PP(T_n = 0) u(x) + F_n^+(x) \]
where $u(x)$ is the (right-continuous) unit-step function with jump at $x = 0$, the probability mass at 0, $\PP(T_n = 0)$, is the outage probability, 
and $F_n^+(x)$ is some right-continuous non-decreasing function of $x$ continuous at $x = 0$, such that 
$\lim_{x \rightarrow +\infty}  F_n^+(x) = 1 - \PP(T_n = 0)$. 

For a given delivery strategy,  we say that no outage occurs whp if $\lim_{n \rightarrow \infty} \PP(T_n = 0) = 0$. 
In addition, we say that a deterministic sequence $\{g^{\rm lb}_n\}$ is achievable if 
$T_n  \overset{\text{whp}}{\geq} g^{\rm lb}_n$. Also, a throughput upper bound whp is defined by 
a deterministic sequence $\{g^{\rm ub}_n\}$ such that 
$T_n  \overset{\text{whp}}{\leq} g^{\rm ub}_n$.  This leads to our definition of achievable throughput scaling laws:

\begin{definition}[Throughput Scaling Law: Achievability] \label{def:throughput}
Given a deterministic sequence $\{g_n^{\rm lb}\}$,  the scaling law $T_n = \Omega(g^{\rm lb}_n)$ is achievable whp if 
there exists a cache placement strategy and delivery protocol such that $T_n  \overset{\text{whp}}{\geq} g^{\rm lb}_n$ 
and $\lim_{n \rightarrow \infty} \PP(T_n = 0) = 0$. \hfill$\lozenge$
\end{definition}

\begin{definition}[Throughput Scaling Law: Converse] \label{def:throughput-converse}
Given a deterministic sequence $\{g_n^{\rm ub}\}$, we say that $T_n = O(g^{\rm ub}_n)$ is a converse throughput scaling law 
whp if for any cache placement strategy and delivery protocol $T_n  \overset{\text{whp}}{\leq} g^{\rm ub}_n$. \hfill$\lozenge$
\end{definition}

Obviously, a tight characterization of the throughput scaling law is obtained when
$T_n = \Omega(g^{\rm lb}_n)$ is achievable whp, and we can exhibit a converse whp $T_n = O(g^{\rm ub}_n)$ such that
when $g_b^{\rm ub} = \Theta( g_n^{\rm lb})$.

\section{Main Results}\label{sec:main}

This section states the main results of this paper.
We first introduce throughput scaling laws of caching wireless \emph{multihop} D2D networks achievable for any popularity distribution in Theorem \ref{thm:main_result} and compare with those of caching wireless \emph{single-hop} D2D networks.
In Theorem \ref{thm:main_result-converse}, we then establish upper bounds on throughput scaling laws for a class of heavy-tailed popularity distributions. In Theorem~\ref{thm:improved_throughput}, we further improve the throughput scaling laws achievable for a Zipf popularity distribution when its exponent is larger than a certain threshold. For ease of exposition, we partition the entire parameter space into five regimes as follows:
\begin{itemize}
\item Regime I: $\alpha-\beta>1$.
\item Regime II: $\alpha-\beta=1$ and $a_1>a_2$.
\item Regime III: $\alpha-\beta=1$ and $a_1\leq a_2$.
\item Regime IV: $\alpha-\beta\in(0,1)$.
\item Regime V: $\alpha-\beta=0$ and $a_1>a_2$.
\end{itemize}
Notice that shifting from Regimes I to V tends to increase the relative caching capability at each node, compared to the library size (recall the relation between $m$ and $M$ in \eqref{eq:relation_parameters}).

The following scaling laws hold \emph{universally} for any  popularity distribution.
\begin{theorem} \label{thm:main_result}
For the caching wireless D2D network defined in Section \ref{sec:PF}, the achievable throughput satisfies whp the scaling laws:
\begin{align} \label{eq:achievable_multihop}
T_{n}=\begin{cases}
0 &\mbox{for Regimes I and II},\\
\Omega(n^{-\frac{1}{2}-\epsilon})&\mbox{for Regime III},\\
\Omega(n^{-\frac{\alpha-\beta}{2}-\epsilon})&\mbox{for Regime IV},\\
\Omega(n^{-\epsilon})&\mbox{for Regime V},
\end{cases}
\end{align}
where $\epsilon > 0$ is arbitrarily small.
\end{theorem}
\begin{proof}
The lower bound for Regimes I and II is trivial.
For the non-trivial part, the proof for Regimes IV and V
is given in Section \ref{subsec:scheme1} and for Regime III in Section \ref{subsec:scheme2}.
\end{proof}

\begin{cor} \label{co:singlehop_ach}
Consider the caching wireless D2D network defined in Section \ref{sec:PF}. If the file delivery is restricted to single-hop transmission, then the achievable throughput satisfies whp the scaling laws:
\begin{align} \label{eq:achievable_singlehop}
T_{n}=\begin{cases}
0 &\mbox{for Regimes I and II},\\
\Omega(n^{-1})&\mbox{for Regime III},\\
\Omega(n^{-(\alpha-\beta)-\epsilon})&\mbox{for Regime IV},\\
\Omega(n^{-\epsilon})&\mbox{for Regime V},
\end{cases}
\end{align}
where $\epsilon > 0$ is arbitrarily small.
\end{cor}
\begin{proof}
The proof is given in Section \ref{subsec:singlehop_ach}.
\end{proof}

\begin{figure}
\begin{center}
\includegraphics[scale=1]{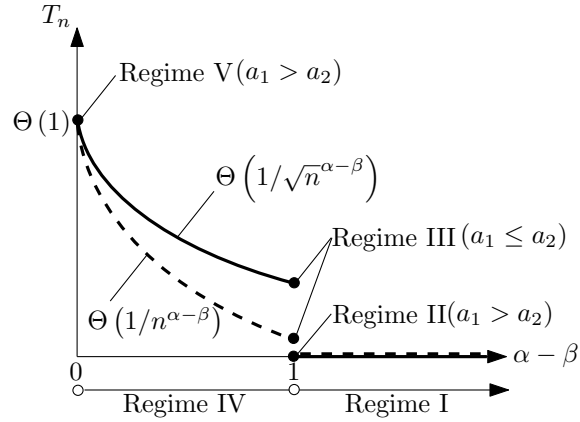}
\end{center}
\caption{Achievable throughput scaling laws in \eqref{eq:achievable_multihop} for caching wireless multihop D2D networks (solid curve) and \eqref{eq:achievable_singlehop} for caching wireless single-hop D2D networks (dashed curve).}
\label{fig:result}
\end{figure}

Fig.~\ref{fig:result} compares the achievable throughput scaling laws of the caching wireless D2D network between multihop and single-hop file deliveries in \eqref{eq:achievable_multihop} and \eqref{eq:achievable_singlehop}, respectively, where we omitted the term $n^{-\epsilon}$ for simplicity. Regimes I and II correspond to the case where the overall cached files in the entire network is strictly less than the number of files in the library, i.e., $Mn<m$. Thus, an outage is inevitable even if a centralized caching is used, which results in $T_n=0$.
As the relative caching capability increases compared to the library size, i.e., $\alpha-\beta$ decreases, each node can find its requested file in the network and thus, a non-zero $T_n$ is achievable for Regimes III and IV.
As it will be clear from the achievability delivery strategies of Section \ref{sec:scheme}, the geometric interpretation of this behavior is as follows:
as $\alpha-\beta$ decreases (i.e., the storage capacity $M$ increases),
the file delivery distance decreases, such that the network achieves larger and larger spatial reuse (multiple links can be active at the same time,
compatibly with the protocol model).
As a result, $T_n$ increases as $\alpha-\beta$ decreases for both \eqref{eq:achievable_multihop} and \eqref{eq:achievable_singlehop}.
Finally when $\alpha=\beta$ (i.e., Regime V), each node can find its requested file from its nearest neighbors. Thus, the delivery distance is
$O(1/\sqrt{n})$ and $T_n = \Theta(1)$ is achievable.

One of the most important facts is that single-hop file delivery is order-optimal only for Regime V.
For almost all parameter space of interest (Regimes III and IV), multihop file delivery significantly improves the throughput
by a factor $\sqrt{\frac{m}{M}}$. Intuitively, spatial reuse is much more effective with multihop transmissions, namely, we can have more
concurrent transmissions in the network. At the same time, the cost of duplicated transmissions by multihop is not very significant comparing
with the gains obtained by the simultaneously active links. It is worthwhile to mention that, for a Zipf popularity distribution with $\gamma < 1$,
our result is well matched with that in \cite{Gitzenisr:13}, even if we use a random caching and a local multihop schemes (see Section \ref{sec:scheme})
rather than a centralized caching and a possibly ``whole-network traversing" multihop schemes presented in \cite{Gitzenisr:13}.
Furthermore, due to the universality of the proposed scheme (random and independent caching), the same throughput scaling laws in Theorem \ref{thm:main_result} and Corollary \ref{co:singlehop_ach} are achievable for random mobile networks since the network caching distribution is invariant with respect to node permutation.

In order to establish upper bounds on throughput scaling laws, we define a class of popularity distributions with the ``heavy tail" property.
\begin{definition}[Heavy-tailed popularity distributions] \label{def:distribution_class}
Define a class of popularity distributions such that, for any $0<c_1<a_1$, there exists $c_2>0$ satisfying that
\begin{align}
\lim_{n\to\infty}\sum_{i=1}^{c_1n^{\alpha}}p_r(i)\leq 1-c_2,
\end{align}
where $c_1$ and $c_2$ are some constants and independent of $n$.
\hfill$\lozenge$
\end{definition}

\begin{lemma}
The Zipf distribution with exponent less than one (i.e., $\gamma \leq 1$) \cite{Breslau:99} satisfies the condition in Definition \ref{def:distribution_class}.
\end{lemma}
\begin{IEEEproof}
Letting $f(n) = \sum_{i=1}^{n}i^{-\gamma}$, we have that
\begin{equation*}
\sum_{i=1}^{c_{1}n^{\alpha}} p_r(i) =
\sum_{i=1}^{c_{1}n^{\alpha}} \frac{i^{-\gamma}}{\sum_{j=1}^{a_1n^\alpha} j^{-\gamma}} =
\frac{f(c_{1}n^{\alpha})}{f(a_{1}n^{\alpha})}.
\end{equation*}
Using the bounds
\begin{equation*}
\int_1^n x^{-\gamma} dx \leq f(n) \leq 1 + \int_{1}^{n} x^{-\gamma} dx, 
\end{equation*}
we have:
\begin{align*}
\lim_{n \rightarrow \infty} \frac{f(c_{1}n^{\alpha})}{f(a_{1}n^{\alpha})} \leq \lim_{n \rightarrow \infty} \frac{c_{1}^{(1-\gamma)}n^{\alpha(1-\gamma)}-\gamma}{a_{1}^{(1-\gamma)}n^{\alpha(1-\gamma)}-1},
\end{align*}where the upper bound converges to $\left ( \frac{c_{1}}{a_{1}} \right )^{(1-\gamma)}$, which is strictly less than 1 since $c_{1} < a_{1}$ and $\gamma < 1$. \end{IEEEproof}

For the above class of popularity distributions,
ignoring a small portion of requests in the tail of the distribution yields a non-vanishing outage probability.
Hence, almost all files in $\mathcal{F}$ should be cached in the network in order to achieve a non-zero $T_n$.
This is the main idea underlying the throughput upper bound in the following theorem.

\begin{theorem} \label{thm:main_result-converse}
Consider the caching wireless D2D network defined in Section \ref{sec:PF} and assume that demands distribution satisfies the condition in Definition \ref{def:distribution_class}.
Then the throughput of any scheme  must satisfy whp the scaling laws:
\begin{align}
T_{n}=\begin{cases}
0 &\mbox{for Regimes I and II},\\
O(n^{-\frac{1}{2}+\epsilon})& \mbox{for Regime III},\\
O(n^{-\frac{\alpha-\beta}{2}+\epsilon}) & \mbox{for Regime IV},\\
O(1/\log n)&\mbox{for Regime V},
\end{cases}
\end{align}
where $\epsilon > 0$ is arbitrarily small.
\end{theorem}
\begin{proof}
The proof is given in Section \ref{subsec:upper_bound1} for Regimes I and II, Section \ref{subsec:upper_bound2} for Regimes III and IV and Section \ref{subsec:upper_regime5} for Regime V.
\end{proof}

For all five regimes, the multiplicative gap between the achievable $T_n$ in Theorem \ref{thm:main_result} and its upper bound in Theorem \ref{thm:main_result-converse} is within $n^{\epsilon}$ for any arbitrarily small $\epsilon>0$.
Therefore, the throughput scaling law depicted in Fig. \ref{fig:result} (solid curve) is order-optimal for the class of heavy-tailed popularity distributions in Definition \ref{def:distribution_class}.
As we will explain in Section \ref{sec:scheme}, in the parameter regimes of interest, such order-optimal throughput scaling is achievable by fully \emph{decentralized random caching uniformly across $\mathcal{F}$}. Similarly, from the following corollary, the throughput scaling law depicted in Fig. \ref{fig:result} (dashed curve) is order-optimal for the class of heavy-tailed popularity distributions in Definition \ref{def:distribution_class} when the file delivery is restricted to single-hop transmission.

\begin{cor} \label{co:singlehop_upper}
Consider the caching wireless D2D network defined in Section \ref{sec:PF} and assume that demands distribution satisfies the condition in Definition \ref{def:distribution_class}.
If the file delivery is restricted to single-hop transmission, then the throughput of any scheme  must satisfy whp the scaling laws:
\begin{align}
T_{n}=\begin{cases}
0 &\mbox{for Regimes I and II},\\
O(n^{-1+\epsilon})& \mbox{for Regime III},\\
O(n^{-(\alpha-\beta)+\epsilon}) & \mbox{for Regime IV},\\
O(1/\log n)&\mbox{for Regime V},
\end{cases}
\end{align}
where $\epsilon > 0$ is arbitrarily small.
\end{cor}
\begin{proof}
The proof is given in Section \ref{subsec:singlehop_upper}.
\end{proof}

\begin{figure}
\begin{center}
\includegraphics[scale=1]{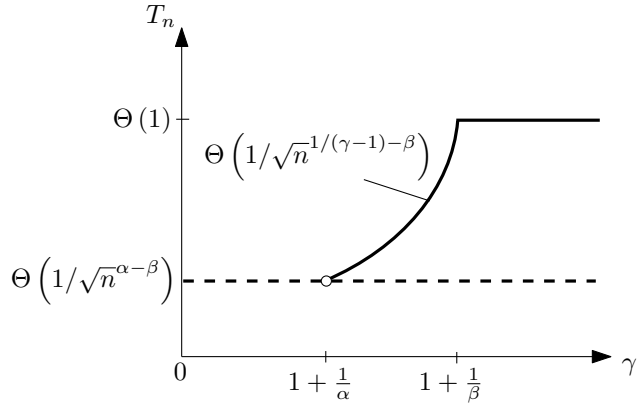}
\end{center}
\caption{Achievable throughput scaling laws in \eqref{eq:improved_throughput} (solid curve) and  \eqref{eq:achievable_multihop} (dashed curve) with respect to $\gamma$  for Regime IV.}
\label{fig:imporved_scaling}
\end{figure}

As the deviation between the request probabilities in the popularity distribution increases (e.g., $\gamma$ increases in a Zipf distribution), 
the condition in Definition \ref{def:distribution_class} may not be satisfied. 
In this case, it can be expected that the throughput scaling law may be improved by a more refined caching strategy, biased towards the files requested with
higher probability. In particular, we consider caching only an appropriately optimized subset of most popular files, while 
guaranteeing that the aggregate ``tail'' probability of the least popular files vanishes, such that we still
get no outage whp. In the following, we demonstrate the above statement for a Zipf popularity distribution 
with $\gamma > 1+\frac{1}{\alpha}$.

\begin{theorem} \label{thm:improved_throughput}
Consider the caching wireless D2D network defined in Section \ref{sec:PF} and assume that the demands follow
a Zipf popularity distribution with exponent $\gamma>1+\frac{1}{\alpha}$. Then the achievable throughput satisfies whp the scaling law:
\begin{align} \label{eq:improved_throughput}
T_{n}= \Omega\left(n^{-\frac{1-\min(1,\beta+1-1/(\gamma-1))}{2}-\epsilon}\right){~~~} \mbox{for Regime IV},
\end{align}
where $\epsilon > 0$ is arbitrarily small.
\end{theorem}
\begin{proof}
The proof is given in Section \ref{sec:improved_throughput}.
\end{proof}

In Fig.~\ref{fig:imporved_scaling}, we compare the improved scaling laws in  \eqref{eq:improved_throughput} and the scaling laws in \eqref{eq:achievable_multihop} for Regime IV, where the term $n^{-\epsilon}$ is omitted. When the demands follow a Zipf popularity distribution, the improved throughput scaling $\Theta(1/\sqrt{n}^{1/(\gamma-1)-\beta})$ is achievable instead of $\Theta(1/\sqrt{n}^{\alpha-\beta})$ in \eqref{eq:achievable_multihop} if $\gamma>1+\frac{1}{\alpha}$ and eventually $\Theta(1)$ scaling is achievable when $\gamma\geq 1+\frac{1}{\beta}$ (see Fig.~\ref{fig:imporved_scaling}).
As we will explain in Section \ref{sec:improved_throughput}, a fully decentralized random caching still achieves the improved throughput scaling laws in Theorem \ref{thm:improved_throughput}, by appropriately reducing the effective library size, i.e., \emph{decentralized random caching uniformly across a subset of popular files}. Namely, in this regime, we can rule out some files from the library, each of which probability is small enough such that an outage does not occur with probability approaching one as $n \rightarrow \infty$.

\textbf{Comparison with the results in \cite{Gitzenisr:13}:} In order to compare our results with these summarized in Table III of \cite{Gitzenisr:13}, we need to let $\alpha \leq 1$ ($n = \Omega(m)$) and $M$ be a constant or $\beta=0$ ($M = \Theta(1)$), then by ignoring the $\epsilon$ in the scaling law exponent, we obtain that $T_n =\Omega\left(\sqrt{\frac{M}{n^{\frac{1}{\gamma-1}}}}\right)=\Omega\left(n^{-\frac{1/(\gamma-1)}{2}}\right) $ under the condition $\gamma > 1 + \frac{1}{\alpha}$ from Theorem \ref{thm:improved_throughput}, which can be either better or worse than the results in \cite{Gitzenisr:13}. For example, if we let $\alpha=1$ and $nM-m = \Theta(1)$, then the throughput in \cite{Gitzenisr:13} is $\Omega\left(\frac{1}{\sqrt{n}}\right)$, which is smaller than $\Omega\left(n^{-\frac{1/(\gamma-1)}{2}}\right)$ for $\gamma > 2$, which is feasible since $\alpha=1$. Remarkably, in this regime, a simple decentralized strategy
consisting of caching the files at random with a uniform distribution over the most popular files,
while discarding the tail of the distribution, can achieve a better throughput than the centralized caching scheme of  \cite{Gitzenisr:13}.
On the other hand,  for $\alpha < 1$, the throughput in \cite{Gitzenisr:13} behaves as $\Omega(1)$,
which is better than $\Omega\left(n^{-\frac{1/(\gamma-1)}{2}}\right)$. In this case, the decentralized random caching strategy might not be sufficient to achieve order optimality under Definition \ref{def:throughput}, i.e., the symmetric rate under no outage.
Whether it is possible to achieve order-optimal throughput scaling with decentralized random caching, allowing for more general
caching distributions (not just uniform over a subset of most probable files) is an interesting question which is left for future research.


\section{Universally Achievable Throughput} \label{sec:scheme}

In this section, we prove Theorem \ref{thm:main_result}.
In particular, we present file placement policies and transmission protocols for Regimes III, IV, and V and analyze their achievable throughput scaling laws.

\subsection{Regimes IV and V} \label{subsec:scheme1}
In this subsection, we prove that
\begin{align} \label{eq:throughput1}
T_n=n^{-\frac{\alpha-\beta}{2}-\epsilon}
\end{align}
is achievable whp for Regimes IV and V, where $\epsilon>0$ is arbitrarily small.

\begin{figure}
\begin{center}
\includegraphics[scale=1]{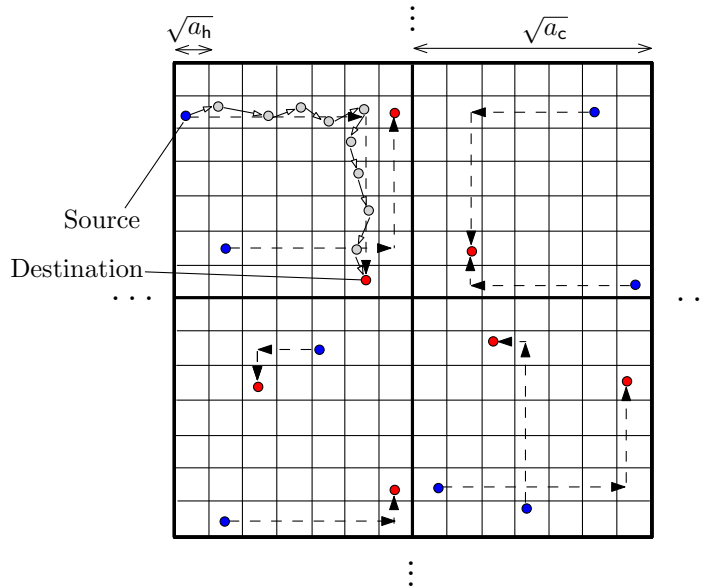}
\end{center}
\caption{The proposed multihop routing protocol for file delivery after the source node selection.}
\label{fig:scheme}
\end{figure}


\subsubsection{File placement policy and transmission protocol} \label{subsubsec:scheme1}
In these regimes, a {\em decentralized} file placement and a {\em local} multihop protocol are proposed as follows.
\\\\
{\bf Decentralized file placement}: Each node $u$ stores $M$ distinct files in its cache,
chosen uniformly at random from the library $\mathcal{F}$, independently of other nodes.
\\\\
{\bf Local multihop protocol}:
We first explain how each node finds its source node having the requested file ({\bf source node selection}):
\begin{itemize}
 \item Divide the entire network into square {\em traffic cells} of area $a_{\sf c}=n^{-\eta}$ for some $\eta\in [0,1)$, where $\eta$ will be determined later on.
 \item Each node chooses one of the nodes having the requested file in the same traffic cell as its {\em source} node. If there are multiple candidates, choose one of them uniformly at random.
\end{itemize}

From Definition~\ref{def:throughput} and the above source node selection, all nodes
should find their source nodes within their own traffic cells whp, in order to achieve a non-zero $T_{n}$. Lemma~\ref{lem:outage} below characterizes such a condition of the area of traffic cell $a_{\sf c}$ (i.e., $\eta$) such as $\eta\in[0,1-(\alpha-\beta))$.

For the ease of exposition, we refer to the pair formed by a node and
its source node as {\em source--destination (SD) pair}. Notice that in our model, each SD pair is located in the same traffic cell while in the conventional wireless ad-hoc network, SD pairs are randomly located over the entire network. Thanks to caching, we can reduce the distance of each SD pair (see Lemma~\ref{lem:outage}). Also, differently from the conventional ad-hoc network model, each node can be a source node of multiple destinations, which make the throughput analysis more complicated (see Lemma~\ref{lemma:num_sources}).

Next, we explain the proposed multihop transmission scheme for the file delivery between $n$ SD pairs, see also Fig.~\ref{fig:scheme}  ({\bf multihop transmission}):
\begin{itemize}
\item Divide each traffic cell into square {\em hopping cells} of area $a_{\sf h}=\frac{2\log n}{n}$.
\item Define the horizontal data path (HDP) and the vertical data
path (VDP) of a SD pair as the horizontal line and the
vertical line connecting a source node to its destination node, respectively.
Each source node transmits the requested file to its destination by first
hopping to the adjacent hopping cells on its HDP and then on its
VDP.\footnote{If a source node and its destination node are in the same hopping cell, then the source node directly transmits to its destination.}
\item Time division multiple access (TDMA) scheme is used with reuse factor $J$
for which each hopping cell is activated only once out of $J$ time slots.
\item A transmitter node in each active hopping cell sends a file (or fragment of a file) to a receiver node in an adjacent hopping cell. Round-robin is used for all transmitter nodes in the same hopping cell.
\end{itemize}

In this scheme, each hopping cell should contain at least one node for relaying as in \cite{GuptaKumar:00,GamalMammenPrabhakarShah:06}, which is satisfied whp since $a_{\sf h}=\frac{2\log n}{n}$ (see Lemma~\ref{lemma:node_dis}~(a)).
\begin{lemma} \label{lemma:node_dis}
The following properties hold whp:
\begin{description}
\item[(a)] Partition the network region $[0,1]\times [0,1]$ into cells of area $\frac{2\log n}{n}$.
Then the number of nodes in each cell is between $1$ and $4\log n$.
\item[(b)] Partition the network region $[0,1]\times [0,1]$ into cells of area $n^{-a}$, where $a\in[0,1)$.
For any $\delta>0$, the number of nodes in each cell is between $(1-\delta)n^{1-a}$ and $(1+\delta)n^{1-a}$. 
\end{description}
\end{lemma}
\begin{proof}
The proofs of first and second properties are given in \cite[Lemma 1]{GamalMammenPrabhakarShah:06} and \cite[Lemma 4.1]{ozgur2007hierarchical}, respectively.
\end{proof}

\begin{lemma} \label{lem:outage}
Suppose Regimes IV and V.  If $\eta\in[0,1-(\alpha-\beta))$, then all nodes are able to find their source nodes within their traffic cells whp.
\end{lemma}
\begin{proof} Let $A_i$ denote the event that node $i$ establishes its source node within its traffic cell, where $i\in[1:n]$. Then, we have:
\begin{align} \label{eq:outage_prob}
\mathbb{P}\left(\cap_{i\in[1:n]}A_i \right)&=1-\mathbb{P}\left(\cup_{i\in[1:n]}A^c_i \right)\nonumber\\
&\geq 1-\sum_{i\in[1:n]}\mathbb{P}\left(A^c_i \right)\nonumber\\
&\overset{\text{whp}}{\geq}1-n \left(\frac{m-M}{m}\right)^{(1-\delta)na_{\sf c}},
\end{align}
where the first inequality follows from the union bound and the second inequality is due to the fact that the number of nodes in each traffic cell is lower bounded by $(1-\delta)na_{\sf c}$ whp (see Lemma~\ref{lemma:node_dis}~(b)) and hence, $\mathbb{P}(A^c_i)\overset{\text{whp}}{\leq} \left(\frac{m-M}{m}\right)^{(1-\delta)na_{\sf c}}$.

Thus, for Regime IV,
\begin{align}
\mathbb{P}\left(\cap_{i\in[1:n]}A_i \right)\overset{\text{whp}}{\geq}1-n\left(\left(1-\frac{a_2}{a_1}\frac{1}{n^{\alpha-\beta}}\right)^{\frac{a_1}{a_2}n^{\alpha-\beta}}\right)^{\frac{a_2(1-\delta)}{a_1}n^{1-\eta-\alpha+\beta}}
\end{align}
and from the fact that
\begin{align} \label{eq:converge_e}
\lim_{n\to\infty}\left(1-\frac{a_2}{a_1}\frac{1}{n^{\alpha-\beta}}\right)^{\frac{a_1}{a_2}n^{\alpha-\beta}}=\frac{1}{e},
\end{align}
$\mathbb{P}\left(\cap_{i\in[1:n]}A_i \right)\to 1$ as $n\to\infty$, since $\eta<1-\alpha+\beta$ is assumed in this lemma.
Similarly, for Regime V,
\begin{align}
\mathbb{P}\left(\cap_{i\in[1:n]}A_i \right)\overset{\text{whp}}{\geq}1-n\left(\frac{a_1-a_2}{a_1}\right)^{(1-\delta)n^{1-\eta}},
\end{align}
which again converges to one as $n\to\infty$, since $a_1>a_2$ for this regime and $\eta<1-\alpha+\beta$ is assumed in this lemma.
In conclusion, all nodes are able to find their source nodes within their traffic cells whp under the condition where $\eta\in[0,1-(\alpha-\beta))$.
\end{proof}

\begin{figure}
\begin{center}
\includegraphics[scale=0.9]{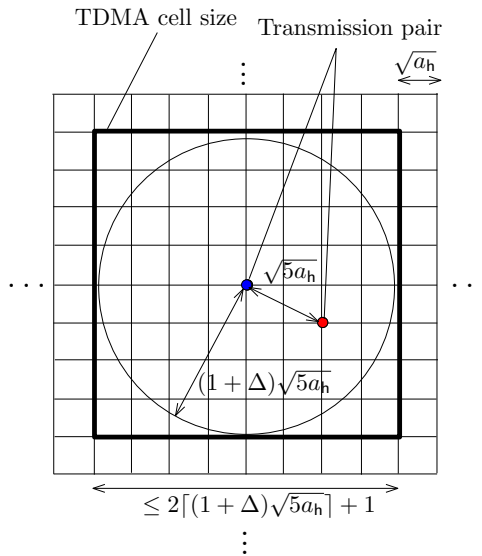}
\end{center}
\caption{TDMA cell size from the protocol model.}
\label{fig:tdma}
\end{figure}

\subsubsection{Achievable throughput} \label{subsubsec:throughput1}
We now show that the proposed scheme in Section~\ref{subsubsec:scheme1} achieves  \eqref{eq:throughput1} whp for Regimes IV and V.
From Lemma~\ref{lem:outage}, we assume $\eta\in[0,1-(\alpha-\beta))$ to achieve a non-zero $T_{n}$ by the proposed scheme in Section~\ref{subsubsec:scheme1}. The following lemmas are instrumental to proof.
\begin{lemma} \label{lemma:aggregate rate}
Suppose Regimes IV and V and assume that $\eta\in[0,1-(\alpha-\beta))$. Let $R_n$ denote the aggregate rate achievable for any hopping cell. If $J\geq \left(2\lceil(1+\Delta)\sqrt{5} \rceil+1\right)^2$, then $R_n=\frac{W}{J}$ is achievable.
\end{lemma}
\begin{proof}
This lemma is a well-known property, e.g., see \cite[Lemma 2]{GamalMammenPrabhakarShah:06}. For completeness, we briefly review proof steps here.
Consider an arbitrary transmission pair consisting of a transmitter node and its receiver node illustrated in Fig.~\ref{fig:tdma}. Clearly, the hopping distance is upper bounded by $\sqrt{5 a_{\sf h}}$ and hence, we choose the transmission range $r=\sqrt{5 a_{\sf h}}$ in the protocol model. Thus, the transmission is successful if there is no node simultaneously transmitting within the distance of $(1+\Delta)\sqrt{5a_{\sf h}}$ from the receiver node. This is satisfied if $J \geq \left(2\lceil(1+\Delta)\sqrt{5} \rceil+1\right)^2$. That is, the aggregate rate of $\frac{W}{J}$ is achievable if  $J\geq\left(2\lceil(1+\Delta)\sqrt{5} \rceil+1\right)^2$. Since this holds for all hopping cells, $R_n=\frac{W}{J}$ is achievable if $J\geq \left(2\lceil(1+\Delta)\sqrt{5} \rceil+1\right)^2$.
\end{proof}

\begin{lemma} \label{lemma:num_sources}
Suppose Regimes IV and V and assume that $\eta\in[0,1-(\alpha-\beta))$. For $\epsilon>0$ arbitrarily small, each node can be a source node of at most $n^{1-\eta-(\alpha-\beta)+\epsilon}$ nodes in its traffic cell whp.
\end{lemma}
\begin{proof}
Let $B_i(k)$ denote the event that node $i$ becomes a source node for less than $k$ nodes.
Denote $n_1=(1+\delta)n^{1-\eta}$.
Then, we have
\begin{align} \label{eq:num_source}
\mathbb{P}\left(\cap_{i\in[1:n]}B_i(k)\right)&=1-\mathbb{P}\left(\cup_{i\in[1:n]}B^c_i(k)\right)\nonumber\\
&\overset{\text {whp}}{\geq}1-n\sum_{j=k}^{n_1}{n_1 \choose j}\left(\frac{M}{m}\right)^j\left(1-\frac{M}{m}\right)^{n_1-j}\nonumber\\
&\geq 1-n\exp\left(-n_1 D\left(\frac{k}{n_1}\bigg\|\frac{M}{m}\right)\right)\nonumber\\
&=1-n\exp\left(-k\log\left(\frac{km}{n_1M}\right)-(n_1-k)\log\left(\frac{m(n_1-k)}{n_1(m-M)}\right)\right)\nonumber\\
&=1-\underbrace{n\exp\left(-k\right)\left(\frac{km}{n_1M}\right)^{-\ln(2)}}_{:=A} \underbrace{\exp\left(-(n_1-k)\right)\left(\frac{m(n_1-k)}{n_1(m-M)}\right)^{-\ln(2)}}_{:=B}
\end{align}
if $\frac{M}{m}<\frac{k}{n_1}<1$, where $D(a\| b)=a\log (\frac{a}{b})+(1-a)\log(\frac{1-a}{1-b})$ denotes the relative entropy for $a,b\in(0,1)$.
Here the first inequality follows from the union bound and holds whp since the number of nodes in each traffic cell is upper bounded by $n_1$ whp from Lemma~\ref{lemma:node_dis} (b), and the second inequality is due to the fact that
for $X\sim B(n,p)$,
\begin{align}
\mathbb{P}(X\geq k)\leq \exp\left(-n D\left(k/n\|p\right)\right)\mbox{ if }p<k/n<1.
\end{align}

First consider Regime IV.
Suppose that $k=n^{\tau}$ for $\tau\in(0,1]$.
Then the condition $\frac{M}{m}<\frac{k}{n_1}<1$ is given by $\frac{a_2(1+\delta)}{a_1}n^{1-\eta-(\alpha-\beta)}<n^{\tau}<(1+\delta)n^{1-\eta}$, which is satisfied as $n$ increases if
\begin{align} \label{eq:num_source_con}
1-\eta-(\alpha-\beta)<\tau\leq1-\eta.
\end{align}
Since $\tau>0$,
\begin{align}
A=n \exp\left(-n^{\tau}\right) \left(\frac{a_1}{a_2(1+\delta)}n^{\tau-1+\eta+(\alpha-\beta)}\right)^{-\ln (2)}
\end{align}
converges to zero as $n$ increases.
Furthermore
\begin{align}
B=&\exp\left(-((1+\delta)n^{1-\eta}-n^{\tau})\right)\left(\frac{a_1n^{\alpha}((1+\delta)n^{1-\eta}-n^{\tau})}{(1+\delta)n^{1-\eta}(a_1n^{\alpha}-a_2n^{\beta})}\right)^{-\ln (2)}
\end{align}
converges to zero as $n$ increases if $\tau\leq 1-\eta$.
In summary, $\mathbb{P}\left(\cap_{i\in[1:n]}B_i(n^{\eta})\right)$ converges to zero as $n$ increases if \eqref{eq:num_source_con} holds.
Therefore, $\mathbb{P}\left(\cap_{i\in[1:n]}B_i(n^{1-\eta-(\alpha-\beta)+\epsilon})\right)\to 0$ as $n\to\infty$ by setting $\tau=1-\eta-(\alpha-\beta)+\epsilon$ for $\epsilon>0$ arbitrarily small, implying that each node becomes a source node of at most $n^{1-\eta-(\alpha-\beta)+\epsilon}$ nodes whp for Regime IV.

Now consider Regime V. Suppose again that $k=n^{\tau}$ for $\tau\in(0,1]$.
Then the condition $\frac{M}{m}<\frac{k}{n_1}<1$ is given by $\frac{a_2(1+\delta)}{a_1}n^{1-\eta}<n^{\tau}<(1+\delta)n^{1-\eta}$, which is satisfied by setting $\tau=1-\eta$ since $a_1>a_2$ for Regime V so that we can find $\delta>0$ satisfying $\frac{a_2(1+\delta)}{a_1}<1$, see  Lemma~\ref{lemma:node_dis}~(b).
For this case, we have
\begin{align}
A=n \exp\left(-n^{1-\eta}\right) \left(\frac{a_1}{a_2(1+\delta)}\right)^{-\ln (2)}
\end{align}
and
\begin{align}
B=\exp\left(-\delta n^{1-\eta}\right)\nonumber\left(\frac{a_1\delta }{(a_1-a_2)(1+\delta)}\right)^{-\ln (2)}.
\end{align}
Hence, $A\to 0$ and $B\to 0$ as $n\to\infty$ since $\eta<1$.
Therefore, each node becomes a source node of at most $n^{1-\eta}$ nodes whp for Regime V.
\end{proof}

Based on Lemma~\ref{lemma:num_sources}, we derive an upper bound on the number of data paths that should be carried by each hopping cell in the following lemma, which is directly related to achievable throughput scaling laws.

\begin{lemma} \label{lemma:num_datapaths}
Suppose Regimes IV and V and assume that $\eta\in[0,1-(\alpha-\beta))$. For $\epsilon>0$ arbitrarily small, each hopping cell is required to carry at
most $n^{\frac{3(1-\eta)}{2}-(\alpha-\beta)+\epsilon}$ data paths whp.
\end{lemma}
\begin{proof}
First consider the number of HDPs that must be carried by an arbitrary hopping cell, denoted by $N_{\sf hdp}$. By assuming that all HDPs of the nodes in the hopping cells located at the same horizontal line pass through the considered hopping cell, we have an upper bound on $N_{\sf hdp}$. Since the total area of these cells is given by
\begin{align} \label{eq:num_nodes_area}
\sqrt{a_{\sf c}a_{\sf h}}&=\sqrt{n^{-\eta}\frac{2\log n}{n}}=n^{\frac{1-\eta}{2}}\frac{1}{\sqrt{2\log n}}\frac{2\log n}{n},
\end{align}
the number of nodes in that area is upper bounded by
\begin{align} \label{eq:num_nodes}
n^{\frac{1-\eta}{2}}\frac{1}{\sqrt{2\log n}}2\log n=n^{\frac{1-\eta}{2}}\sqrt{2\log n}
\end{align}
whp from Lemma~\ref{lemma:node_dis} (a).
Moreover, each of these nodes may become a source node of multiple nodes within the same traffic cell. Therefore, from Lemma~\ref{lemma:num_sources} and \eqref{eq:num_nodes}
\begin{align}
N_{\sf hdp}&\overset{\text{whp}}{\leq} n^{1-\eta-(\alpha-\beta)+\epsilon'}n^{\frac{1-\eta}{2}}\sqrt{2\log n} \nonumber\\
&=n^{\frac{3(1-\eta)}{2}-(\alpha-\beta)+\epsilon'}\sqrt{2\log n}
\end{align}
for $\epsilon'>0$ arbitrarily small.
The same analysis holds for VDPs.
In conclusion, each hopping cell carries at most $n^{\frac{3(1-\eta)}{2}-(\alpha-\beta)+\epsilon}$ data paths whp for $\epsilon>0$ arbitrarily small, which completes the proof.
\end{proof}

We are now ready to prove that \eqref{eq:throughput1} is achievable whp for Regimes IV and V.
Let $\epsilon'>0$ be an arbitrarily small constant satisfying that $1-(\alpha-\beta)-\epsilon'>0$, which is valid for Regimes IV and V since $\alpha-\beta\in[0,1)$.
Then set $\eta=1-(\alpha-\beta)-\epsilon'$, which determines the size of each traffic cell. From Lemma~\ref{lem:outage}, every node can find its source node within its traffic cell whp. From Lemma~\ref{lemma:aggregate rate}, setting $J=\left(2\lceil(1+\Delta)\sqrt{5} \rceil+1\right)^2$, each hopping cell is able to achieve the aggregate rate of
\begin{equation} \label{eq:agg_rate}
R_n=W/\left(2\lceil(1+\Delta)\sqrt{5} \rceil+1\right)^2.
\end{equation}
Furthermore, from Lemma \ref{lemma:num_datapaths}, the number of data paths that each hopping cell needs to perform is upper bounded by
\begin{equation} \label{eq:num_datapaths}
n^{\frac{3(1-\eta)}{2}-(\alpha-\beta)+\epsilon'}=n^{\frac{\alpha-\beta}{2}+\frac{5}{2}\epsilon'}
\end{equation}
whp, where we used $\eta=1-(\alpha-\beta)-\epsilon'$.

Since each hopping cell serves multiple data paths using round-robin fashion, each data path is served with a rate of at least \eqref{eq:agg_rate} divided by  \eqref{eq:num_datapaths} whp.
Therefore, an achievable per-node throughput is given by
\begin{align}
T_n
& = \frac{W}{\left(2\lceil(1+\Delta)\sqrt{5} \rceil+1\right)^2}n^{-\frac{\alpha-\beta}{2}-\frac{5}{2}\epsilon'}\geq n^{-\frac{\alpha-\beta}{2}-\epsilon}
\end{align}
whp for $\epsilon>0$ arbitrarily small.
In conclusion, \eqref{eq:throughput1} is achievable whp for Regimes IV and V.

\subsection{Regime III}  \label{subsec:scheme2}
In this subsection, we prove that
\begin{align} \label{eq:throughput2}
T_n= n^{-\frac{1}{2}-\epsilon}
\end{align}
is achievable whp assuming that $\alpha-\beta=1$ and $a_1=a_2$, where $\epsilon>0$ is arbitrarily small.
Hence the same $T_n$ is also achievable whp for $\alpha-\beta=1$ and any $a_1\leq a_2$, which corresponds to Regime III.

From now on, assume that $\alpha-\beta=1$ and $a_1=a_2$.
For this case, the total number of files that can be stored by $n$ nodes (i.e., the total number of files stored in the network) is exactly the same as the number of files in the library (i.e., $nM=m$).
We propose a {\em centralized} file placement and a {\em globally} multihop protocol as follows.
\\\\
{\bf Centralized file placement}: 
It can be seen that a distributed file placement might result in an outage, as seen from the analysis in Lemma~\ref{lem:outage}. Instead, we employ a simple centralized file placement for which all distinct $m$ files (in the library) are randomly stored in the total memories of $n$ nodes. Hence, the network can contain all $m$ files, thus being able to avoid an outage.
\\\\
{\bf Globally multihop protocol}: As explained before, the traffic cell should be equal to the entire network (i.e., $\eta = 0$ in Section~\ref{subsec:scheme1}), in order to avoid an outage. Namely, $n$ SD pairs are located over the entire network. Hence, we can expect the same scaling result with the conventional wireless ad-hoc network in  \cite{GuptaKumar:00}, namely, no caching gain is expected.
\\

We briefly explain how to achieve \eqref{eq:throughput2} whp, since the procedures of proof are almost similar to Regimes IV and V. Similarly to Lemma~\ref{lemma:num_sources}, we can show that each node is able to be a source node of at most $n^{\epsilon}$ nodes whp for $\epsilon>0$ arbitrarily small. Then, following the analysis in Section~\ref{subsubsec:throughput1}, we can easily prove that \eqref{eq:throughput2} is achievable whp for Regime III.

\subsection{Single-Hop File Delivery} \label{subsec:singlehop_ach}

In this subsection, we prove Corollary \ref{co:singlehop_ach}.
First consider Regimes IV and V.
We apply the same file placement and source node selection policy described in Section \ref{subsubsec:scheme1}.
Then Lemma \ref{lem:outage} holds guaranteeing no outage whp by setting $\eta=1-(\alpha-\beta)-\epsilon'$, where $\epsilon'>0$ be an arbitrarily small constant satisfying that $1-(\alpha-\beta)-\epsilon'>0$.
Consider the file delivery. Instead of multihop routing within each traffic cell, each source directly transmits the required file to its destination within each traffic cell. Then, from the same analysis in Lemma \ref{lemma:aggregate rate}, each traffic cell achieves the aggregate rate of $R_n=\frac{W}{\left(2\lceil(1+\Delta)\sqrt{5} \rceil+1\right)^2}$ by  TDMA between traffic cells with reuse factor $\left(2\lceil(1+\Delta)\sqrt{5} \rceil+1\right)^2$.
Since there are at most $(1+\delta)n^{(\alpha-\beta)+\epsilon'}$ nodes in each traffic cell whp (Lemma \ref{lemma:node_dis} (b)), the rate of $\frac{R_n}{1+\delta}n^{-(\alpha-\beta)-\epsilon'}$ is achievable whp for each file delivery.
Therefore, $T_n=n^{-(\alpha-\beta)-\epsilon}$ is achievable whp for Regimes IV and V, where $\epsilon>0$ is arbitrarily small.

Now consider Regime III. As the same reason in \ref{subsec:scheme2}, we assume $\alpha-\beta=1$  and $a_1=a_2$, and then apply  the same file placement and source node selection policy described in Section \ref{subsec:scheme2}, which guarantees no outage.
Then, from the direct file delivery by time-sharing between $n$ SD pairs, $T_n=\frac{1}{n}$ is achievable whp for Regime III.

\section{Converse} \label{sec:upper_bound}

In this section, we prove the upper bounds in Theorem \ref{thm:main_result-converse} assuming that the popularity distribution satisfies the condition in Definition \ref{def:distribution_class}.
We first introduce the following technical lemma.
\begin{lemma} \label{lemma:chernoff} Let $X$ follow a binomial distribution with parameters $l$ and $p$, i.e., $X \sim B(l,p)$. Then, for $k\in[0:lp]$,
\begin{align}  \label{eq:chernoff}
\mathbb{P}(X\leq k) \leq \exp\left(-\frac{1}{2p}\frac{(lp-k)^2}{l}\right).
\end{align} 
\end{lemma}
\begin{proof}
The proof follows immediately from the Chernoff bound.
\end{proof}
\subsection{Regimes I and II} \label{subsec:upper_bound1}
We introduce the following lemma, which demonstrates that a non-vanishing outage probability is inevitable for Regimes 1 and 2 even if centralized caching were allowed.
Therefore, a non-vanishing outage probability implied by
Lemma  \ref{lemma:outage1} yields that $T_{n} = 0$ whp for Regimes I and II.

\begin{lemma} \label{lemma:outage1}
Suppose Regimes I and II.
Let $N_{{\sf out},1}$ denote the number of nodes that they cannot find their requested files in the entire network. Then, we have $N_{{\sf out},1}\geq c_3 n$ whp for some constant $c_3>0$ independent of $n$.
\end{lemma}
\begin{proof}
The total number of files that are able to be stored by the entire network is given by $nM=a_2n^{1+\beta}$.
Hence the probability that each node cannot find its requested file in the entire network is lower bounded by
\begin{align} \label{eq:upper_outage1}
1-\sum_{i=1}^{a_2n^{1+\beta}}p_r(i):=p_{{\sf out},1}.
\end{align}
Then, for  $\mu\in[0,p_{{\sf out},1}]$, we have
\begin{align} \label{eq:Pr_N_out1}
\mathbb{P}(N_{{\sf out},1}\geq \mu n)&\overset{(a)}{\geq} \sum_{i=\mu n}^{n}{n \choose i} p_{{\sf out},1}^i(1-p_{{\sf out},1})^{n-i}\nonumber\\
&\geq 1-\sum_{i=1}^{\mu n}{n \choose i} p_{{\sf out},1}^i(1-p_{{\sf out},1})^{n-i}\nonumber\\
&\overset{(b)}{\geq}  1-\exp\left(-\frac{(p_{{\sf out},1}-\mu)^2}{2 p_{{\sf out},1}}n\right),
\end{align}
where $(a)$ follows from \eqref{eq:upper_outage1} and the fact that each node requires a file independent of other nodes and $(b)$ follows from Lemma \ref{lemma:chernoff}.
Here, the condition $\mu\in[0,p_{{\sf out},1}]$ is required to apply Lemma \ref{lemma:chernoff}.

Now consider $p_{{\sf out},1}$ defined in \eqref{eq:upper_outage1}.
Notice that $a_2n^{1+\beta}< a_1n^{\alpha}$ as $n\to\infty$ for both regimes because $\alpha-\beta>1$ for Regime I and $\alpha-\beta=1$ and $a_1>a_2$ for Regime II.
Hence, from Definition~\ref{def:distribution_class}, $\lim_{n\to\infty}p_{{\sf out},1}\geq c_4$ for some constant $c_4>0$ independent of $n$.
Then setting $\mu=\frac{c_4}{2}$ in \eqref{eq:Pr_N_out1}, which satisfies $\mu\in[0,p_{{\sf out},1}]$ as $n\to\infty$, yields that
$\mathbb{P}\left(N_{{\sf out},1}\geq \frac{c_4}{2}n\right)\to 1$ as $n\to\infty$.
Therefore, $N_{{\sf out},1}\geq c_3 n$ whp for some constant $c_3>0$ independent of $n$.
\end{proof}

\subsection{Regimes III and IV} \label{subsec:upper_bound2}

The key ingredient to establish the upper bounds in Theorem \ref{thm:main_result-converse} for Regimes III and IV is to characterize the minimum 
distance for file transmission that a non-zero fraction of SD pairs must go through, 
which is given in Lemma~\ref{lemma:outage2} below. Then, as a consequence of the protocol model which does not allow concurrent transmission 
within a  circle of radius $(1+\Delta)r$ around each intended receiver, we are able to determine how many SD pairs can be simultaneously 
activate at a given time slot, which is directly related to the desired throughput upper bounds.

\begin{figure}[t!]
\begin{center}
\includegraphics[scale=0.9]{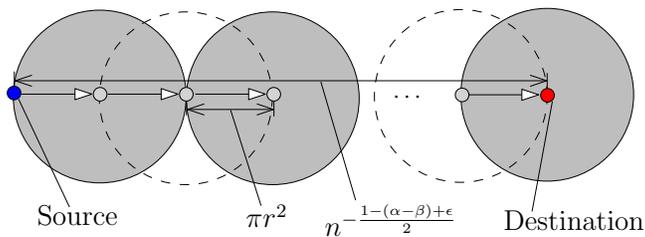}
\end{center}
\vspace{-0.15in}
\caption{A lower bound on the exclusive area occupied by the multihop transmission of a SD pair with distance $n^{-\frac{1-(\alpha-\beta)+\epsilon}{2}}$.}
\label{fig:upper}
\vspace{-0.2in}
\end{figure}

\begin{lemma} \label{lemma:outage2}
Suppose Regimes III and IV.
For $\epsilon>0$ arbitrarily small, let $N_{{\sf out},2}$ denote the number of nodes that they cannot find their requested files within the distance of $n^{-\frac{1-(\alpha-\beta)}{2}-\epsilon}$ from their positions.
Then, we have $N_{{\sf out},2}\geq c_5n$ whp for some constant $c_5>0$ independent of $n$.
\end{lemma}
\begin{proof}
Let $\epsilon'>0$ be an arbitrarily small constant satisfying that $1-(\alpha-\beta)+\epsilon'\in [0,1)$, which is valid for Regimes III and IV since $\alpha-\beta\in(0,1]$.
For simplicity, denote $\zeta=\frac{1-(\alpha-\beta)+\epsilon'}{2}$.
Let $N_{\sf file}$ be the total number of files that are able to be stored by the area of radius $n^{-\zeta}$.
From Lemma \ref{lemma:node_dis} (b), the number of nodes in that area is upper bounded by $(1+\delta)n^{1-2\zeta}$ whp because $2\zeta\in[0,1)$.
Hence $N_{\sf file}\leq (1+\delta)n^{1-2\zeta}M=a_2(1+\delta)n^{\alpha-\epsilon'}$ whp.
Then the probability that each node cannot find its requested file within the radius of $n^{-\zeta}$ is lower bounded by
\begin{align} \label{eq:upper_outage2}
1-\sum_{i=1}^{N_{\sf file}}p_r(i)\overset{\text{whp}}{\geq} 1-\sum_{i=1}^{a_2(1+\delta)n^{\alpha-\epsilon'}}p_r(i):=p_{{\sf out},2}.
\end{align}
Then similarly to \eqref{eq:Pr_N_out1}, we have
\begin{align} \label{eq:Pr_N_out2}
\mathbb{P}(N_{{\sf out},2}\geq \mu n)\overset{\text{whp}}{\geq} 1-\exp\left(-\frac{(p_{{\sf out},2}-\mu)^2}{2 p_{{\sf out},2}}n\right)
\end{align}
for  $\mu\in[0,p_{{\sf out},2}]$.
From Definition~\ref{def:distribution_class}, $\lim_{n\to\infty}p_{{\sf out},2}\geq c_6$ for some constant $c_6>0$ independent of $n$.
More specifically, we can apply Definition~\ref{def:distribution_class} because $a_2(1+\delta)n^{\alpha-\epsilon'}< a_1 n^{\alpha}$ as $n\to\infty$.
Hence setting $\mu=\frac{c_6}{2}$ in \eqref{eq:Pr_N_out2}, which satisfies $\mu\in[0,p_{{\sf out},2}]$ as $n\to\infty$, yields that
$\mathbb{P}\left(N_{{\sf out},2}\geq \frac{c_6}{2}n\right)\to 1$ as $n\to\infty$.
Therefore, $N_{{\sf out},2}\geq c_5 n$ whp for some constant $c_5>0$ independent of $n$.
\end{proof}

Based on Lemma~\ref{lemma:outage2}, we can prove that the throughput of any scheme must satisfy
\begin{align} \label{eq:upper_bound_3_4}
T_n \overset{\text{whp}}{\leq} n^{-\frac{\alpha-\beta}{2}+\epsilon}
\end{align}
for Regimes III and IV, where $\epsilon>0$ is arbitrarily small.
Specifically, from Lemma~\ref{lemma:outage2}, for $\epsilon'>0$ arbitrarily small, there are at least $c_{5}n$ SD pairs whose distances are larger than $n^{-\frac{1-(\alpha-\beta)}{2}-\epsilon'}$ whp, where $c_{5}>0$ is some constant and independent of $n$.
Then, we restrict only on the delivery of the requests of such SD pairs, obtaining clearly an upper bound on the per-node throughput.
First, we consider the exclusive area (i.e., the area to prohibit the transmission for other SD pairs) occupied by the multihop transmission of a SD pair with distance $n^{-\frac{1-(\alpha-\beta)}{2}-\epsilon'}$. In order to obtain a lower bound on such area, we assume that $\Delta=0$ and each receiver node is located at the distance of $r$ from its transmitter node along with the SD line (see Fig.~\ref{fig:upper}). Then, the exclusive area is lower bounded (i.e., only taking the shaded areas in Fig.~\ref{fig:upper}) such as
\begin{equation}
\frac{2\pi r^2 n^{-\frac{1-(\alpha-\beta)}{2}-\epsilon'}}{2r} = \pi r n^{-\frac{1-(\alpha-\beta)}{2}-\epsilon'}.
\end{equation}
Hence, the maximum number of SD pairs guaranteeing a rate of $W$ over the entire network of a unit area is upper 
bounded by $\frac{1}{\pi r}n^{\frac{1-(\alpha-\beta)}{2}+\epsilon'}$ whp. 
As a result, the sum throughput $S_n$ (summing the rate of all users) is upper bounded by 
$S_n \overset{\text{whp}}{\leq} \frac{W}{\pi r}n^{\frac{1-(\alpha-\beta)}{2}+\epsilon'}$. Notice that for a given sum throughput $S_n$, 
the symmetric per-user rate is trivially upper bounded by $T_n \leq S_n/n$. Hence, we have
\begin{align} \label{eq:upper_bound_r}
T_n \leq  \frac{S_n}{n} \overset{\text{whp}}{\leq} \frac{W}{\pi r}n^{\frac{-1-(\alpha-\beta)}{2}+\epsilon'}.
\end{align}
That the above bound on $T_n$ increases as $r$ decreases. 
On the other hand, it was shown in \cite[Section V]{GuptaKumar:00} that the absence of isolated nodes is a necessary condition for a 
non-zero $T_{n}$ requiring that
\begin{align} \label{eq:connectivity}
r\overset{\text{whp}}{\geq} c_7\sqrt{\log n/n}
\end{align}
for some constant $c_7>0$ independent of $n$. Therefore, from \eqref{eq:upper_bound_r} and \eqref{eq:connectivity}, we have an upper bound on the per-node throughput as
\begin{align}
T_n&\overset{\text{whp}}{\leq}\frac{W}{\pi c_7}n^{-\frac{\alpha-\beta}{2}-\frac{\log\log n}{2\log n}+\epsilon'}\nonumber\\
&\leq n^{-\frac{\alpha-\beta}{2}+\epsilon}
\end{align}
for $\epsilon>0$ arbitrarily small. In conclusion, the upper bound in \eqref{eq:upper_bound_3_4} holds whp for Regimes III and IV.

\subsection{Regime V} \label{subsec:upper_regime5}
In this subsection, we prove that the throughput of any scheme must satisfy
\begin{align} \label{eq:upper_bound_5}
T_n \overset{\text{whp}}{\leq} \frac{c_8}{\log n}
\end{align}
for Regime V, where $c_8>0$ is some constant independent of $n$.
The following lemma shows that at least a constant fraction of nodes have to download their requested files from other nodes, which will be used as the key ingredient to prove the upper bound in \eqref{eq:upper_bound_5}.

\begin{lemma} \label{lemma:outage3}
Suppose Regime V.
Let $N_{{\sf out},3}$ denote the number of nodes that they cannot find their requested files in their own cache memories. Then, we have $N_{{\sf out},3}\geq c_9 n$ whp for some constant $c_9>0$ independent of $n$.
\end{lemma}
\begin{proof}
Similar to the proof in Lemmas \ref{lemma:outage1} and \ref{lemma:outage2}, we have
\begin{align} \label{eq:Pr_N_out3}
\mathbb{P}(N_{{\sf out},3}\geq \mu n)\geq 1-\exp\left(-\frac{(p_{{\sf out},3}-\mu)^2}{2 p_{{\sf out},3}}n\right)
\end{align}
for  $\mu\in[0,p_{{\sf out},3}]$, where $p_{{\sf out},3}=1-\sum_{i=1}^{a_2 n^{\alpha}}p_r(i)$.
Since $a_2<a_1$ for Regime V, $\lim_{n\to\infty}p_{{\sf out},3}\geq c_{10}$ for some constant $c_10>0$ independent of $n$ from Definition \eqref{def:distribution_class}. Hence setting $\mu=\frac{c_{10}}{2}$ in \eqref{eq:Pr_N_out3} yields that
$\mathbb{P}\left(N_{{\sf out},3}\geq \frac{c_{10}}{2}n\right)\to 1$ as $n\to\infty$.
Therefore, $N_{{\sf out},3}\geq c_9 n$ whp for some constant $c_9>0$ independent of $n$.
\end{proof}

From Lemma \ref{lemma:outage3}, a non-vanishing fraction of nodes have to download their requested files from other nodes and, as a result, \eqref{eq:connectivity} should be satisfied for successful file delivery, see \cite[Section V]{GuptaKumar:00}.
From the protocol model, then, the rate of each file delivery is upper bounded by $W$ bits/sec/Hz and there are at most $\frac{1}{\pi c_7^2}\frac{n}{\log n}$ concurrent file deliveries in the network whp, from the bound in \eqref{eq:connectivity}. Therefore, 
$S_n \overset{\text{whp}}{\leq} \frac{W}{\pi c_7^2}\frac{n}{\log n}$ and $T_n \leq \frac{S_n}{n} \overset{\text{whp}}{\leq} \frac{W}{\pi c_7^2}\frac{1}{\log n}$.
In conclusion the upper bound in \eqref{eq:upper_bound_5} holds whp for Regime V.

\subsection{Single-Hop File Delivery} \label{subsec:singlehop_upper}
In this subsection, we prove Corollary \ref{co:singlehop_upper}.
For Regimes I and II, Lemma \ref{lemma:outage1} still holds, resulting that $T_n=0$ whp for these regimes.
Also, the same argument in Section \ref{subsec:upper_regime5} holds, resulting that \eqref{eq:upper_bound_5} whp for Regime V.

Now consider Regimes III and IV. From Lemma \ref{lemma:outage2}, if the file delivery for SD pairs with distance at least $n^{-\frac{1-(\alpha-\beta)}{2}-\epsilon}$ is restricted to single-hop transmission,  the exclusive area occupied by  each of those SD pairs is lower bounded by
\begin{align}
\pi n^{-(1-(\alpha-\beta))-\epsilon'}
\end{align}
whp for $\epsilon'>0$ arbitrarily small.
Then, as the same analysis in Section \ref{subsec:upper_bound2}, we have $S_n \overset{\text{whp}}{\leq} \frac{W}{\pi}n^{1-(\alpha-\beta)+\epsilon'}$, 
resulting that $T_n \overset{\text{whp}}{\leq} n^{-(\alpha-\beta)+\epsilon}$ for $\epsilon>0$ arbitrarily small for these regimes.

\section{Improved Achievable Throughput} \label{sec:improved_throughput}
In this section, we prove Theorem \ref{thm:improved_throughput} by assuming that user demands follow a Zipf popularity distribution with exponent $\gamma>1+\frac{1}{\alpha}$.

\subsection{File Placement and Delivery}
Similar to the case of Regime IV in Section \ref{subsec:scheme1}, i.e., $\alpha-\beta\in(0,1)$, a {\em distributed} file placement and a {\em local} multihop protocol are performed. In order to describe the proposed file placement, let $\epsilon_c>0$ be an arbitrarily small constant satisfying that
\begin{align} \label{eq:condition_e_c}
\beta+1-\frac{1}{\gamma}-\epsilon_c>0,
\end{align}
which is valid because
$\beta+1-\frac{1}{\gamma-1}>1-(\alpha-\beta)>0$, where the first inequality holds from the assumption $\gamma>1+\frac{1}{\alpha}$ and the second inequality holds for Regime IV since $\alpha-\beta\in(0,1)$.
Then, define
\begin{align} \label{eq:def_n_2}
n_2=n^{1-\min(1,\beta+1-1/(\gamma-1))+\epsilon_c/2}
\end{align}
and let $\mathcal{F}_{\sf{sub}}\subseteq \mathcal{F}$
denote the subset of the first (most probable) $(Mn_2)$ files in the library.
During the file placement phase, each node stores $M$ distinct files in its cache, chosen uniformly at random from $\mathcal{F}_{\sf{sub}}$ independently of other nodes.

During the file delivery phase, the same local mulithop described in Section \ref{subsec:scheme1} is performed.
To determine the size of each traffic cell, we set
\begin{align}  \label{eq:eta_value}
\eta=\min\left(1,\beta+1-\frac{1}{\gamma-1}\right)-\epsilon_c,
\end{align}
which is valid since $\eta\in(0,1)$ from \eqref{eq:condition_e_c}.
Then the number of nodes in each traffic cell is upper bounded by
\begin{align} \label{eq:upper_n_nodes}
(1+\delta)n_2n^{\epsilon_c/2}
\end{align}
whp and lower bounded by
\begin{align}  \label{eq:lower_n_nodes}
(1-\delta)n_2n^{\epsilon_c/2}
\end{align}
whp from Lemma \ref{lemma:node_dis} (b).

\subsection{Achievable Throughput}
In this subsection, we prove that
\begin{align} \label{eq:improved_throughput_rewrite}
T_{n} = n^{-\frac{1-\min(1,\beta+1-1/(\gamma-1))}{2}-\epsilon}
\end{align}
is achievable whp for Regime IV, where $\epsilon > 0$ is arbitrarily small.
The overall procedure is similar to the case of Regime IV in Section \ref{subsec:scheme1}.
In the following, we first show that all nodes can find their required files within their traffic cells whp by setting $\eta$ as in \eqref{eq:eta_value}.

\begin{lemma} \label{lemma:outage_3}
Suppose Regime IV and $\eta=\min\left(1,\beta+1-\frac{1}{\gamma-1}\right)-\epsilon_c$. Then all nodes are able to find their sources within their traffic cells 
whp.
\end{lemma}
\begin{proof}
Denote $P=\sum_{i=1}^{Mn_2}p_r(i)$, where the definition of $n_2$ is given by \eqref{eq:def_n_2}.
For $i\in[1:n]$, denote $\mathcal{N}_i\subseteq[1:n]$ as the set of nodes in the traffic cell that node $i$ is included and $A_i$ as the event that node $i$ establishes its source node in $\mathcal{N}_i$.
Then the outage probability $\mathbb{P}(A_i^c)$ is given by
\begin{align}
\mathbb{P}(A_i^c)&=\mathbb{P}(\mbox{node $i$ requests $f_i\in\mathcal{F}_{\sf sub}$})\mathbb{P}(f_i\notin\cup_{j\in \mathcal{N}_i}\mathcal{M}_j|\mbox{node $i$ requests $f_i\in\mathcal{F}_{\sf sub}$})\nonumber\\
&+\mathbb{P}(\mbox{node $i$ requests $f_i\notin\mathcal{F}_{\sf sub}$})\nonumber\\
&=P\left(\frac{Mn_2-M}{Mn_2}\right)^{|\mathcal{N}_i|}+(1-P)\nonumber\\
&\overset{\text {whp}}{\leq}P\left(1-\frac{1}{n_2}\right)^{(1-\delta)n_2n^{\epsilon_c/2}}+(1-P),
\end{align}
where $|\mathcal{N}_i|$ denotes the cardinality of $\mathcal{N}_i$.
Here, the second equality holds since each node $i$ stores $M$ distinct files in its local memory $\mathcal{M}_i$, chosen uniformly at random from $\mathcal{F}_{\sf{sub}}$ independently of other nodes and the inequality holds from \eqref{eq:lower_n_nodes}.

Then, following the analysis in \eqref{eq:outage_prob}, we have:
\begin{align} \label{eq:outage_prob_3}
\mathbb{P}\left(\cap_{i\in[1:n]}A_i \right)&\overset{\text {whp}}{\geq} 1-n\left(P\left(1-\frac{1}{n_2 }\right)^{(1-\delta)n_2n^{\epsilon_c/2}}+(1-P)\right)\nonumber\\
&\geq 1-n\left(1-\frac{1}{n_2}\right)^{(1-\delta)n_2n^{\epsilon_c/2}}-n(1-P)\nonumber\\
&=1-n\left(\left(1-\frac{1}{n_2}\right)^{n_2}\right)^{(1-\delta)n^{\epsilon_c/2}}-n(1-P).
\end{align}

From \eqref{eq:converge_e} and the fact that $n_2\to\infty$ as $n\to\infty$, the term $n\left(\left(1-\frac{1}{n_2}\right)^{n_2}\right)^{(1-\delta)n^{\epsilon_c/2}}$ in \eqref{eq:outage_prob_3} convergeges to zero as $n$ increases. Furthermore,
\begin{align} \label{eq:outage_prob_4}
n(1-P)&\overset{(a)}{=}n\left(\frac{\sum_{i=Mn_2+1}^{m}i^{-\gamma}}{\sum_{i=1}^{m}i^{-\gamma}}\right)\nonumber\\
&\overset{(b)}{\leq}n\left(\frac{\int_{Mn_2}^{m}x^{-\gamma}d_x}{\int_{1}^{m}x^{-\gamma}d_x}\right)\nonumber\\
&\overset{(c)}{=}\frac{n^{\alpha(1-\gamma)+1}-n^{(\beta+1-\min(1,\beta+1-1/(\gamma-1))+\epsilon_c/2)(1-\gamma)+1}}{n^{\alpha(1-\gamma)}-1},
\end{align}
where $(a)$ follows from the definition of $P$, $(b)$ follows because $\sum_{i=a+1}^{b}i^{-\gamma}\leq \int_{a}^{b}x^{-\gamma}dx$ and $\sum_{i=a}^{b}i^{-\gamma}\geq \int_{a}^{b}x^{-\gamma}dx$, and $(c)$ follows from the definition of $n_2$.
Notice that $n^{\alpha(1-\gamma)}$ and $n^{\alpha(1-\gamma)+1}$ in \eqref{eq:outage_prob_4} converge to zero as $n$ increases since $\gamma> 1+\frac{1}{\alpha}$.
Also,
\begin{align}
n^{(\beta+1-\min(1,\beta+1-1/(\gamma-1))+\epsilon_c/2)(1-\gamma)+1}&\leq n^{(\beta+1-(\beta+1-1/(\gamma-1))+\epsilon_c/2)(1-\gamma)+1}\nonumber\\
&=n^{\frac{\epsilon_c/2}{1-\gamma}}
\end{align}
converges to zero as $n$ increases because $\frac{\epsilon_c/2}{1-\gamma}<0$.
Therefore, the term $n(1-P)$ in \eqref{eq:outage_prob_3} also converges to zero as $n$ increases.

In conclusion, from \eqref{eq:outage_prob_3}, $\mathbb{P}\left(\cap_{i\in[1:n]}A_i \right)$ converges to zero as $n$ increases.
\end{proof}

As proved in Lemma \ref{lemma:outage_3}, we set $\eta=\min\left(1,\beta+1-\frac{1}{\gamma-1}\right)-\epsilon_c$ from now on, which determines the size of each traffic cell guaranteeing no outage at all nodes whp.
Notice that Lemma \ref{lemma:aggregate rate} holds regardless of the file popularity distribution.
Hence, a non-vanishing aggregate rate is achievable for any hopping cell by TDMA between hopping cells with some constant reuse factor.
We then derive the same statement in Lemma \ref{lemma:num_sources} in the following lemma.

\begin{lemma} \label{lemma:num_sources2}
Suppose Regime IV and $\eta=\min\left(1,\beta+1-\frac{1}{\gamma-1}\right)-\epsilon_c$. Then each node can be a source node of at most $n^{\epsilon_c}$ nodes in its traffic cell whp.
\end{lemma}
\begin{proof}
Let $B_i(k)$ denote the event that node $i$ becomes a source node for less than $k$ nodes.
From the same analysis in \eqref{eq:num_source}, we have
\begin{align}
\mathbb{P}\left(\cap_{i\in[1:n]}B_i(k)\right)&\overset{\text {whp}}{\geq}1-n\sum_{j=k}^{(1+\delta)n_2n^{\epsilon_c/2}}{(1+\delta)n_2n^{\epsilon_c/2} \choose j}\left(\frac{1}{n_2}\right)^j\left(1-\frac{1}{n_2}\right)^{(1+\delta)n_2n^{\epsilon_c/2}-j}\nonumber\\
&\geq 1-n\exp\left(-(1+\delta)n_2n^{\epsilon_c/2} D\left(\frac{k}{(1+\delta)n_2n^{\epsilon_c/2}}\bigg\|\frac{1}{n_2}\right)\right)\nonumber\\
&=1-n\exp\left(-k\log\left(\frac{k}{(1+\delta)n^{\epsilon_c/2}}\right)\right)\nonumber\\
&\cdot\exp\left(-((1+\delta)n_2n^{\epsilon_c/2}-k)\log\left(\frac{(1+\delta)n_2n^{\epsilon_c/2}-k}{(1+\delta)n_2n^{\epsilon_c/2}-(1+\delta)n^{\epsilon_c/2}}\right)\right)\nonumber\\
&=1-\underbrace{n\exp(-k)\left(\frac{k}{(1+\delta)n^{\epsilon_c/2}}\right)^{-\ln (2)}}_{:=C}\nonumber\\
& \cdot\underbrace{\exp(-((1+\delta)n_2n^{\epsilon_c/2}-k))\left(\frac{(1+\delta)n_2n^{\epsilon_c/2}-k}{(1+\delta)n_2n^{\epsilon_c/2}-(1+\delta)n^{\epsilon_c/2}}\right)^{-\ln (2)}}_{:=D}
\end{align}
if $\frac{1}{n_2}<\frac{k}{(1+\delta)n_2n^{\epsilon_c/2}}<1$, where $D(a\| b)=a\log (\frac{a}{b})+(1-a)\log(\frac{1-a}{1-b})$ denotes the relative entropy for $a,b\in(0,1)$.

Suppose that $k=n^{\epsilon_c}$.
Then the condition $\frac{1}{n_2}<\frac{k}{(1+\delta)n_2n^{\epsilon_c/2}}<1$ is satisfied because $(1+\delta)n^{\epsilon_c/2}<n^{\epsilon_c}<(1+\delta)n^{1-\min(1,\beta+1-1/(\gamma-1))+\epsilon_c}$.
Furthermore, we have
\begin{align}
C&= n \exp(-n^{\epsilon_c})\left(\frac{n^{\epsilon_c/2}}{1+\delta}\right)^{-\ln (2)}\to 0 \mbox{ as } n \to \infty.
\end{align}
Similarly, from the definition of $n_2$ in \eqref{eq:def_n_2},
\begin{align}
D&= \exp(-((1+\delta)n^{1-\min(1,\beta+1-1/(\gamma-1))+\epsilon_c}-n^{\epsilon_c}))\nonumber\\
&\cdot\left(\frac{(1+\delta)n^{1-\min(1,\beta+1-1/(\gamma-1))+\epsilon_c}-n^{\epsilon_c}}{(1+\delta)n^{1-\min(1,\beta+1-1/(\gamma-1))+\epsilon_c}-(1+\delta)n^{\epsilon_c/2}}\right)^{-\ln (2)}\to 0 \mbox{ as } n \to \infty
\end{align}
because $((1+\delta)n^{1-\min(1,\beta+1-1/(\gamma-1))+\epsilon_c}-n^{\epsilon_c})\to\infty$ as $n\to\infty$ and $\frac{(1+\delta)n^{1-\min(1,\beta+1-1/(\gamma-1))+\epsilon_c}-n^{\epsilon_c}}{(1+\delta)n^{1-\min(1,\beta+1-1/(\gamma-1))+\epsilon_c}-(1+\delta)n^{\epsilon_c/2}}\to 1$ as $n\to\infty$.
Therefore, $\mathbb{P}\left(\cap_{i\in[1:n]}B_i(n^{\epsilon_c})\right)\to 1$ as $n\to\infty$, meaning that each node becomes a source node of at most $n^{\epsilon_c}$ nodes whp.
\end{proof}

\begin{lemma} \label{lemma:num_datapaths2}
Suppose Regime IV and $\eta=\min\left(1,\beta+1-\frac{1}{\gamma-1}\right)-\epsilon_c$. For $\epsilon>0$ arbitrarily small, each hopping cell is required to carry at
most $n^{\frac{1-\min(1,\beta+1-1/(\gamma-1))}{2}+\epsilon}$  data paths whp.
\end{lemma}
\begin{proof}
Let $N_{\sf hdp}$ denote the number of HDPs that must be carried by an arbitrary hopping cell.
From the same analysis in \eqref{eq:num_nodes_area} and \eqref{eq:num_nodes} and Lemma \ref{lemma:num_sources2}, we have
\begin{align}
N_{\sf hdp}&\overset{\text{whp}}{\leq}
=n^{\frac{1-\min(1,\beta+1-1/(\gamma-1))}{2}+\epsilon_c}\sqrt{2\log n}.
\end{align}
The same analysis holds for VDPs.
In conclusion, each hopping cell carries at most $n^{\frac{1-\min(1,\beta+1-1/(\gamma-1))}{2}+\epsilon}$ data paths whp for $\epsilon>0$ arbitrarily small, which completes the proof.
\end{proof}

We are now ready to prove that \eqref{eq:improved_throughput_rewrite} is achievable whp for Regime IV.
From Lemma~\ref{lemma:outage_3}, every node can find its source node within its traffic cell whp. From Lemma~\ref{lemma:aggregate rate}, setting $J=\left(2\lceil(1+\Delta)\sqrt{5} \rceil+1\right)^2$, each hopping cell is able to achieve the aggregate rate in \eqref{eq:agg_rate}.
Furthermore, from Lemma \ref{lemma:num_datapaths2}, the number of data paths that each hopping cell needs to perform is upper bounded by
\begin{equation} \label{eq:num_datapaths2}
n^{\frac{1-\min(1,\beta+1-1/(\gamma-1))}{2}+\epsilon'}
\end{equation}
whp for $\epsilon'$ arbitrarily small.
Therefore, an achievable per-node throughput is given by at least \eqref{eq:agg_rate} divided by  \eqref{eq:num_datapaths2} whp.
In conclusion, \eqref{eq:improved_throughput_rewrite} is achievable whp for Regime IV.

\section{Concluding Remarks}\label{sec:conc}

We considered a wireless ad-hoc network in which nodes have cached information from a library of possible files. For such network, we proposed an order-optimal caching policy (i.e., file placement policy) and multihop transmission protocol for a broad class of heavy-tailed popularity distributions including a Zipf distribution with exponent less than one. Interestingly, we showed that a distributed uniform random caching is order-optimal {for the parameter regimes of interest} as long as the total number of files in the library is less than the overall caching memory size in the network. i.e., $\alpha-\beta\in(0,1]$. Also, it was shown that a multihop transmission provides a significant throughput gain over one-hop direct transmission as in the conventional wireless ad-hoc networks.
As a future work, the complete characterization of the optimal throughput scaling laws for  this network with random demands following a
Zipf distribution with an arbitrary exponent $\gamma$ (in particular, with $\gamma \geq 1$) remains to be determined.
In this regime, decentralized uniform random caching over a subset of most probable files is generally not order-optimal, and gains can be achieved by
more refined random decentralized caching policies. Whether these can achieve the same scaling laws of the deterministic centralized strategy of
\cite{Gitzenisr:13} in all regimes remains also to be seen.


\end{document}

%% file: macros.tex
\long\def\comment#1{}

\newfont{\bbb}{msbm10 scaled 700}

\newfont{\bb}{msbm10 scaled 1100}

\newcommand{\PP}{\mbox{\bb P}}

\newcommand{\FF}{\mbox{\bb F}}

\newcommand{\Ac}{{\cal A}}

\newcommand{\Ec}{{\cal E}}
\newcommand{\Fc}{{\cal F}}
\newcommand{\Gc}{{\cal G}}

\newcommand{\Uc}{{\cal U}}

\newcommand{\fsf}{{\sf f}}

\newcommand{\Gsf}{{\sf G}}

\newcommand{\Psf}{{\sf P}}

\newcommand{\be}{\begin{equation}}
\newcommand{\ee}{\end{equation}}
\newcommand{\bea}{\begin{eqnarray}}
\newcommand{\eea}{\end{eqnarray}}

\def\fsf{ {\sf f}}